\documentclass[%
 aip,
 amsmath,amssymb,
 reprint,%
 onecolumn,
 nofootinbib
]{revtex4-2}

\renewcommand*{\thefootnote}{\fnsymbol{footnote}}

\usepackage{graphicx}% Include figure files
\usepackage{dcolumn}% Align table columns on decimal point
\usepackage{bm}% bold math
\usepackage[utf8]{inputenc}
\usepackage[T1]{fontenc}
\usepackage{mathptmx}
\usepackage{etoolbox}

\usepackage{tikz}

\makeatletter
\def\@email#1#2{%
 \endgroup
 \patchcmd{\titleblock@produce}
  {\frontmatter@RRAPformat}
  {\frontmatter@RRAPformat{\produce@RRAP{*#1\href{mailto:#2}{#2}}}\frontmatter@RRAPformat}
  {}{}
}%
\makeatother

\newcommand{\N}{\mathbb{N}}
\newcommand{\Z}{\mathbb{Z}}
\newcommand{\R}{\mathbb{R}}
\newcommand{\C}{\mathbb{C}}
\newcommand{\T}{\mathbb{T}}
\newcommand{\Hi}{\mathcal{H}}
\newcommand{\set}[1]{\left\{#1\right\}}
\newcommand{\eu}{\mathrm{e}}
\newcommand{\iu}{\mathrm{i}}
\newcommand{\di}{\mathrm{d}} 
\newcommand{\sub}[1]{_{\mathrm{#1}}}
\newcommand{\Id}{\mathbf{1}}

\newcommand{\ket}[1]{\left| #1 \right\rangle}
\newcommand{\bra}[1]{\left\langle #1 \right|}

\newcommand{\bk}{\mathbf{k}}
\newcommand{\bx}{\mathbf{x}}
\newcommand{\bR}{\mathbf{R}}
\newcommand{\bG}{\mathbf{G}}

\DeclareMathOperator{\Tr}{Tr}

\DeclareMathOperator{\Ran}{Ran}
\DeclareMathOperator{\odeg}{1-deg}
\DeclareMathOperator{\tdeg}{3-deg}

\usepackage{amsthm}

\theoremstyle{plain}
\newtheorem{theorem}{Theorem}[section]
\newtheorem{lemma}[theorem]{Lemma}
\newtheorem{corollary}[theorem]{Corollary}
\newtheorem{proposition}[theorem]{Proposition}

\theoremstyle{definition}

\newtheorem{remark}[theorem]{Remark}

\newtheorem{assumption}[theorem]{Assumption}

\usepackage{enumerate}
\usepackage{enumitem}

\begin{document}

\title[Topology vs localization in synthetic dimensions]{Topology vs localization in synthetic dimensions}
% Force line breaks with \\
\author{Domenico Monaco}%
 \homepage{Email: \href{mailto:domenico.monaco@uniroma1.it}{domenico.monaco@uniroma1.it}}
 \affiliation{Dipartimento di Matematica, ``Sapienza'' Universit\`{a} di Roma\\ Piazzale Aldo Moro 5, 00185 Rome, Italy
%\\This line break forced with \textbackslash\textbackslash
}%
\author{Thaddeus Roussign\'e}
 \homepage{Email: \href{mailto:thaddeus.roussigne@ens-paris-saclay.fr}{thaddeus.roussigne@ens-paris-saclay.fr}}
 \affiliation{\'Ecole Normale Sup\'erieure Paris--Saclay\\ 4, Avenue des Sciences, 91190 Gif-sur-Yvette, France}

\date{\today, revised manuscript [first submitted on 10 October 2022]}% It is always \today, today,

\begin{abstract}
Motivated by recent developments in quantum simulation of synthetic dimensions, e.g.\ in optical lattices of ultracold atoms, we discuss here $d$-dimensional periodic, gapped quantum systems for $d \le 4$, with focus on the topology of the occupied energy states. We perform this analysis by asking whether the spectral subspace below the gap can be spanned by smooth and periodic Bloch functions, corresponding to localized Wannier functions in position space. By constructing these Bloch functions inductively in the dimension, we show that if they are required to be orthonormal then in general their existence is obstructed by the first two Chern classes of the underlying Bloch bundle, with the second Chern class characterizing in particular the 4-dimensional situation. If the orthonormality constraint is relaxed, we show how $m$ occupied energy bands can be spanned by a Parseval frame comprising at most $m+2$ Bloch functions.
\end{abstract}

\maketitle

\section{Introduction}

Engineering quantum simulation devices is a very active field of research in experimental physics, which promises to shed light on complex condensed matter phenomena using table-top experimental setups. In particular, several proposals have been brought forward to emulate so-called \emph{synthetic dimensions} in quantum simulators, which allow to probe features of the system as if it could ``move'' along more spatial dimensions than it actually possesses. This is achieved by carefully designing and coupling extra degrees of freedom for the system, which can be modelled mathematically as these extra dimensions and moreover offer a very flexible tunability. Remarkably, these efforts have allowed to experimentally study 4-dimensional (4D) quantum and condensed-matter systems, in particular for what pertains their topological properties, which were previously only envisioned theoretically \cite{zhang2001four, qi2008topological}. Such 4D systems (with three spatial dimensions and one synthetic dimension, or two and two respectively) have been realized in optical lattices of ultracold atoms \cite{price2015four}, in topological charge pumps \cite{lohse2018exploring}, in photonic waveguides \cite{zilberberg2018photonic}, as well as with acoustic waves \cite{chen2021creating} and with twisted bilayer phononic lattices \cite{rosa2021topological}; we refer the reader to Ref.\  \onlinecite{ozawa2019topological} for a recent review regarding the thriving research on topological quantum matter in synthetic dimensions.

The prototypical example of a topological phenomenon in condensed matter physics is arguably the integer quantum Hall effect (IQHE) \cite{klitzing1980new}, where the quantization of the transverse (Hall) conductivity of a 2-dimensional (2D) electron gas, subject to a perpendicular magnetic field and driven out of equilibrium by an in-plane electric field, is explained by relating the integer value it assumes (in appropriate units) to the \emph{first Chern number} of the underlying \emph{Bloch bundle} of occupied energy states \cite{thouless1982quantized, graf2007aspects}. Arguing by analogy, a similar non-linear response effect has been proposed to occur in 4D (time-reversal invariant) topological insulators \cite{zhang2001four, qi2008topological}, where the underlying topological integer is instead the \emph{second Chern number}. In this paper, we intend to present these topological numbers from a different perspective, namely as \emph{obstructions} to the existence of an orthonormal basis of smooth and periodic Bloch functions which span the fibers of the Bloch bundle. We also discuss the situation in which the orthonormality constraint is removed for this generating set of Bloch functions: there, we will see that a \emph{Parseval frame} of smooth and periodic Bloch functions always exists, irrespective of the vanishing or non-vanishing of the first two Chern classes, and we will also characterize the minimal number of required generators which constitute the frame in the general, topologically non-trivial situation.

As the expert reader will notice, our results are more or less well known in differential geometry, once the Bloch functions are understood as (orthonormal or spanning) sections of a Hermitian vector bundle --- the Bloch bundle --- over the $d$-dimensional torus --- the Brillouin torus \cite{husemoller}. In our presentation, however, we will present \emph{algorithmic, constructive proofs} of the existence of the required smooth and periodic Bloch functions, aiming at concreteness and in view of possible applications in numerical condensed matter physics, where these tools are extensively used (see Ref.\ \onlinecite{marzari2012maximally} and references therein).

\subsection{Setting}

In condensed matter physics, crystalline solids in $d$-dimensions have a configuration space which is invariant by a Bravais lattice of translations $\Gamma \simeq \Z^d$: the latter identification occurs after having determined an appropriate basis which generates the lattice directions. This translation invariance corresponds to a conserved ``quantum number'', the crystal (or Bloch) momentum $\bk$: this is itself determined up to traslations in the dual Bravais lattice $\Gamma^* \simeq 2\pi\Z^d$ whose elements $\bG$ are determined by the condition $\bG \cdot \bR \in 2\pi \Z$ for all $\bR \in \Gamma$. Effectively, this constrains the crystal momentum on a $d$-dimensional torus $\T^d := \R^d / \Gamma^*$, called the Brillouin torus. Moreover, these lattice translations are required to be unitarily represented on the Hilbert space of the quantum particle, and the Hamiltonian of the system is required to commute with these translation operators. This implies that the Hamiltonian itself preserves the crystal momentum $\bk$, and therefore it makes sense to discuss how it acts on wavefunctions which have a well-defined momentum and depend only on the degrees of freedom in the fundamental (Wigner--Seitz) cell of the lattice $\Gamma$. This action of the Hamiltonian is denoted by $H(\bk) = H(\bk + \bG)$, $\bG \in \Gamma^*$. Mathematically, it is obtained from the original Hamiltonian in position space (say, a Schr\"{o}dinger-type operator, or a tight-binding, discrete approximation thereof) by the Bloch--Floquet transform \cite{kuchment2016overview}. The theory can accomodate also magnetic translations \cite{zak1964magnetic}, provided the magnetic flux per unit cell is commensurate with respect to the quantum of magnetic flux. The use of a modified Bloch--Floquet transform, also called Bloch--Floquet--Zak transform, is at times more mathematically convienent, but leads to operators which are only unitarily equivalent and not equal when the crystal momentum is shifted by a dual-lattice translation. At any rate, even in this situation, the dual lattice representation can be modified in order to restore exact periodicity in the crystal momentum. For a discussion on this and related topics, we refer the reader to Ref.s\ \onlinecite{monaco2015symmetry, monaco2018optimal} and references therein.

The type of systems we will be then interested in are topological insulators in class A, according to the Altland--Zirnbauer--Cartan label in Kitaev's ``periodic table'' of topological quantum matter \cite{heinzner2005symmetry, kitaev2009periodic, ryu2010topological}. The IQHE, as well as its 4D analogue mentioned above, enter in this classification in $d=2$ and $d=4$, respectively. These systems are described by the following Assumption, which is verified in many sensible models \cite{monaco2018optimal}.

\begin{assumption}[Class-A topological insulator] \label{assum:top_ins}
The operators $H(\bk)$ are self-adjoint operators on some Hilbert space $\Hi$ (typically the $L^2$-space over the Wigner--Seitz cell of the lattice $\Gamma$), uniformly bounded from below. The resolvent map
\[ \T^d \ni \bk \mapsto \left[ H(\bk) - \iu \, \Id \right]^{-1} \in \mathcal{B}(\Hi) \]
is assumed to be $C^\infty$-smooth, and to take values in compact operators on $\Hi$. Moreover, we assume that there exist $\mu \in \R$ and $g > 0$ such that, for all $\bk \in \T^d$, the interval $[\mu - g, \mu + g]$ does not intersect the spectrum of $H(\bk)$. This interval is then called the \emph{spectral gap} of the Hamiltonian.
\end{assumption}

The spectral gap assumption allows to define the \emph{spectral projection} $P(\bk) = P(\bk)^2 = P(\bk)^* \in \mathcal{B}(\Hi)$ of the Hamiltonian $H(\bk)$ onto the energy levels below the gap, e.g.\ by the Riesz formula
\[ P(\bk) := \frac{\iu}{2\pi} \oint_C \di z \, \left[ H(\bk) - z \, \Id \right]^{-1}\,. \]
In the above, $C$ is a contour in the complex energy plane which intersects the real energy axis at $\mu$ and below the bottom of the spectrum of $H(\bk)$, with the latter choice being performed uniformly in $\bk$. It can be argued \cite{panati2013bloch} that, under our Assumption \ref{assum:top_ins}, the spectral projections $P(\bk)$ also depend smoothly and periodically on $\bk$. Its rank --- the dimension of its range --- is then constant in $\bk$, and will be denoted by $m \in \N$: it counts the number of occupied energy levels, below the spectral gap. The range of the spectral projection $P(\bk)$ consists of the span of the corresponding eigenfunction of the Hamiltonian $H(\bk)$, namely (the periodic parts of) the \emph{Bloch functions}. With an abuse of terminology and for lack of a better name, we will refer to any vector in the range of $P(\bk)$ as a Bloch function%
\footnote{Sometimes these vectors are termed quasi-Bloch functions, but we couldn't find a sufficiently common and accepted terminology.}.

Informally, the questions we will address are the following:
\begin{enumerate}
 \item Is it possible to span the range of the projections $P(\bk)$ with \emph{smooth and periodic} Bloch functions $\phi_a(\bk) \in \Hi$?
 \item Can the vectors $\phi_a$ be chosen to be \emph{orthonormal}?
 \item If not, what is the \emph{minimal number} of smooth and periodic vectors which is needed to span the range of the projection $P(\bk)$?
\end{enumerate}
The interest in these questions stems from the importance of the position-space counterparts (Bloch--Floquet anti-transforms) of the vectors $\phi_a$, which are called (\emph{composite}) \emph{Wannier functions} $w_a$. Specifically, \emph{localized} Wannier functions, corresponding to \emph{smooth} periodic Bloch functions, are a valuable tool in (numerical) condensed matter physics, with application ranging from the justification of the tight-binding approximation to the interpolation of numerical data in electronic structure \cite{marzari2012maximally}. In general, smoothness and orthonormality of periodic Bloch functions (that is, localization and orthonormality of Wannier functions) compete with one another, and typically one can only enforce one of the conditions at the expense of the other. The origin of this competition lies in the non-trivial topology of the \emph{Bloch bundle}, a vector bundle on the Brillouin torus $\T^d$ which can be naturally associated to the family of projections $P(\bk)$ for ranging $\bk \in \T^d$. For example, in the language of differential geometry, Bloch functions correspond to \emph{sections} of the Bloch bundle, and the existence of an \emph{orthonormal basis} of \emph{smooth} and periodic Bloch functions is equivalent to the (\emph{topological}) \emph{triviality} of the bundle \cite{panati2007triviality, monaco2015symmetry}. 

With this ``dictionary'' at hand, one can realize what are the topological conditions required to span the spectral projections with (non-)orthonormal smooth and periodic Bloch functions \cite{husemoller}. These conditions are often formulated in terms of so-called characteristic classes of the Bloch bundle, which are cohomology classes on the torus associated with the bundle. The relevant characteristic classes in the present setting are \emph{Chern classes}. These are defined by means of the \emph{Berry curvature} of the Bloch bundle, which is the following operator-valued 2-form on the torus $\T^d$:
\begin{equation} \label{eqn:Berry_curv} 
F = F(P) := \frac{1}{2\pi\iu} \, P \, \di P \wedge \di P \, P \equiv \frac{1}{2} \sum_{1 \le \mu < \nu \le d} F_{\mu\nu} \, \di k_\mu \wedge \di k_\nu, \quad \text{with} \quad F_{\mu\nu}(\bk) := \frac{1}{2\pi\iu} P(\bk) \, \big[ \partial_{k_\mu} P(\bk), \partial_{k_\nu} P(\bk) \big] P(\bk). 
\end{equation}
The \emph{$n$-th Chern form} of the Bloch bundle is then the differential form of degree $2n$ on $\T^d$ which is determined inductively by the following identity \cite{milnor1974characteristic, nash2011topology}:
\[ c_0(P) := 1, \quad c_n(P) := \frac{1}{n} \sum_{i=1}^{n} (-1)^{i-1} c_{n-i}(P) \wedge \Tr\big(F^{\wedge i}\big), \quad n \ge 1, \]
with $F^{\wedge i} := F \wedge \stackrel{\text{($i$ times)}}{\cdots} \wedge\, F$. In particular
\begin{equation} \label{eqn:c12}
c_1(P) =  \Tr(F), \quad c_2(P) = \frac{1}{2} \left[ \Tr(F) \wedge \Tr(F) - \Tr\big(F \wedge F\big) \right]\,.
\end{equation}
These differential forms can be shown to be closed ($\di c_n(P) = 0$), and therefore $c_n(P)$ gives rise to a cohomology class in the de Rham cohomology of the torus: this is the \emph{$n$-th Chern class} $[c_n(P)] \in H^{2n}(\T^d)$. Notice in particular that $[c_{n}(P)]=0$ if $2n$ exceeds the dimension $d$ of the Brillouin torus, so there are at most $\lceil d/2 \rceil$ non-trivial Chern classes (including the $0$-th one). The topological triviality of the Bloch bundle can then be tested by asking whether these classes vanish in cohomology (i.e.\ whether the bundle is \emph{Chern-trivial}): in turn, on the torus, the cohomological triviality of $[c_n(P)]$ can be reformulated by the easier condition that all the \emph{$n$-th Chern numbers} vanish. The latter are defined as follows: Choose $I = \{ i_1 < i_2 < \ldots < i_n \} \subset \set{1, \ldots, d}$ an ordered collection of $n$ labels among the $d$ coordinates of the torus $\T^d$, and define
\begin{equation} \label{eqn:ChernNumbers}
c_n^{I}(P) := \int_{\T^n_I} c_n(P)
\end{equation}
where $\T^n_I$ is the $n$-dimensional sub-torus of $\T^d$ obtained by freezing the value of the $(n-d)$ coordinates different from $k_{i_1}, \ldots, k_{i_n}$. These numbers can be argued to be \emph{integers}:
\[ c_n^{I}(P) \in \Z. \]
While it is always true that a trivial bundle is also Chern-trivial, the converse implication is in general false; the two notions turn out to be equivalent only if the rank of the projections $m$ equals $1$, or if $m>1$ and the dimension of the Brillouin torus $d$ is sufficiently small. We will see also from our main result that, in the latter case, $d \le 4$ suffices.

\subsection{Results}

After this brief \textit{excursus} in vector-bundle theory, we are ready to collect all the previous considerations into our main result. Motivated by the discussion about quantum simulation devices with synthetic dimensions presented at the beginning of the Introduction, we address class-A topological insulators in dimension $d \le 4$. We have then

\begin{theorem} \label{thm:main}
Let $P(\bk)$, $\bk \in \T^d$, be the family of spectral projections of a $d$-dimensional class-A topological insulator as in Assumption \ref{assum:top_ins}. Assume that $d \le 4$, and denote by $m$ the (constant) rank of the spectral projections. Then the following hold.

\begin{enumerate}
 \item \label{item:orthonormal} There exists an \emph{orthonormal basis} of smooth and periodic Bloch functions $\set{\phi_1(\bk), \ldots, \phi_m(\bk)}$ spanning $P(\bk)$ if and only if the Bloch bundle is Chern-trivial, that is, 
  \begin{equation} \label{eqn:Chern_trivial}
  c_1^I(P) = 0 \quad \text{for all} \quad I = \set{i < j} \subset \set{1,\ldots, 4} \qquad \text{\emph{and}} \qquad c_2^{\set{1,2,3,4}}(P) = 0.
  \end{equation}
 \item \label{item:Parseval} There always exists a \emph{Parseval frame} of smooth and periodic Bloch functions $\set{\phi_1(\bk), \ldots, \phi_M(\bk)}$ spanning $P(\bk)$, with
 \begin{itemize}
  \item $M \le m+1$ if $d \le 3$;
  \item $M \le m+2$ if $d=4$.
 \end{itemize}
\end{enumerate}
\end{theorem}

Let us recall that a \emph{Parseval frame} for a projection $P$ on an Hilbert space is a collection of vectors $\set{\phi_i}_{i \in \mathcal{I}}\subset \Ran P$ such that
\[ P \phi = \phi \quad \Longrightarrow \quad \phi = \sum_{i \in \mathcal{I}} \left\langle \phi_i, \phi \right\rangle \, \phi_i, \]
that is, the vectors span the range of the projection and the Parseval identity holds, but the vectors themselves are not required to be orthonormal. 

\begin{remark} \label{rmk:comments}
We make a few comments on the statement of Theorem \ref{thm:main}.
\begin{enumerate}
 \item The parts of the result that concern $d$-dimensional class-A topological insulators for $d \le 3$ have already appeared elsewhere \cite{cornean2019parseval}; our new contribution here concerns the 4D case. Since the proof is ``inductive'' in the dimension, below we will review and elaborate on the construction in $d \le 3$ to pave the way for the presentation of the case $d=4$. A similar existence result, also covering dimensions $d \le 4$ and concerned with the analytic rather than smooth setting (therefore with exponentially rather than polynomially localized Wannier functions), can be found in Ref.\ \onlinecite{auckly2018parseval} (see also Ref.s\ \onlinecite{denittis2011exponentially, denittis2020erratum}).
 \item We stress once again that, while results of this type concerning spanning sections of a vector bundle are known in the differential geometry community \cite{husemoller, avis1979quantum}, we will provide a \emph{construction} of the required smooth and periodic Bloch functions. In this construction, the topological obstruction will appear naturally and will be manifested through (possibly non-vanishing) Chern numbers. This constructive proof goes through some ``deformation arguments'' (namely through homotopy theory), which we try to make as explicitly as possible.
 \item The number of conditions listed in \eqref{eqn:Chern_trivial}, depending on the dimension, is as follows:
 \begin{itemize}
  \item[$d=1$] --- no condition: both the first and the second Chern forms vanish automatically because they are 2- and 4-forms on the 1-dimensional torus, respectively;
  \item[$d=2$] --- one first Chern number $c_1^{\set{1,2}}(P) \in \Z$ needs to vanish; the second Chern form vanishes identically again for dimensional reasons;
  \item[$d=3$] --- three first Chern numbers $c_1^{\set{1,2}}(P),\, c_1^{\set{2,3}}(P),\, c_1^{\set{1,3}}(P) \in \Z$ need to vanish; the second Chern form vanishes identically again for dimensional reasons;
  \item[$d=4$] --- six first Chern numbers \emph{and} one second Chern number need to vanish. Moreover, if $m=1$ then the second Chern class always vanishes, thus reducing the number of conditions to be checked: this will also become apparent in our proof.
 \end{itemize}
 \item In 2D, the only first Chern number that arises is the integer responsible for the quantization of the Hall conductivity in the IQHE \cite{graf2007aspects}. In 4D, a similar role is played by the second Chern number \cite{zhang2001four}. Moreover, in the jargon of topological quantum matter \cite{kennedy2015homotopy}, the first Chern number is a \emph{strong invariant} in $d=2$, as well as the second Chern number in $d=4$; in $d=3$ and $d=4$, the first Chern numbers are instead \emph{weak invariants}. The terminology for strong invariants reflects the fact that the corresponding Chern class is of top degree in the appropriate dimension, while the first Chern class is a lower-dimensional object in $d \ge 3$, leading to the term ``weak invariant''.
 \item In the second part of the statement, the bounds on the number of vectors in the Parseval frame for the spectral projection are optimal: that is, in dimension $d \le 3$ one needs $M=m+1$ vectors in the frame as soon as one of the first Chern numbers is non-zero, and in dimension $4$ one needs $M=m+2$ vectors in the frame as soon as the second Chern number is non-zero. In 4D, there could be situations in which the second Chern number vanishes while one of the first Chern numbers does not (for example by virtue of some extra symmetry of the quantum system); in this case, one could have a Parseval frame of only $m+1$ vectors. This will be apparent in any case from the proof. In view of the first part of the statement, if all first and second Chern numbers vanish one can then make the Parseval frame an orthonormal basis, that is, one can even choose $M=m$.
\end{enumerate}
\end{remark}

\section{1D case and parallel transport}

We start the proof of the main Theorem \ref{thm:main} by considering the 1-dimensional (1D) case. As emphasized in Remark \ref{rmk:comments}, in this case the construction of an orthonormal basis of smooth and periodic Bloch functions for the spectral projections $P(k)$ is topologically unobstructed. We can see this as follows. Define the \emph{parallel transport} unitaries $T(k) \in \mathcal{U}(\Hi)$ as the solution to the following operator-valued Cauchy problem:
\begin{equation} \label{eqn:parallel}
\begin{cases}
\iu \, \partial_k T(k) = K(k) \, T(k), & K(k) := \iu \big[ \partial_k P(k), P(k) \big], \\
T(0) = \Id. 
\end{cases}
\end{equation}
The unitaries $T(k)$ depend $C^\infty$-smoothly on $k \in \R$, and they intertwine the spectral projections, that is,
\[ P(k) = T(k) \, P(0) \, T(k)^*. \]
The parallel trasport unitaries, however, lack in general periodicity: they rather satisfy the telescopic relation $T(k+2\pi) = T(k)\, T(2\pi)$, and in general $T(2\pi) \ne T(0) = \Id$ --- one says that, after a loop around the Brillouin torus $\T^1 \simeq \R / 2\pi \Z$, parallel-transported vectors pick up a \emph{holonomy} $T(2\pi)$. However, one can always write $T(2\pi) = \eu^{2 \pi \iu X}$ for some bounded and self-adjoint operator $X \in \mathcal{B}(\Hi)$ which commutes with $P(0) = P(2\pi)$ in view of the above intertwining property. See e.g.\ Ref.\ \onlinecite{cornean2016construction} for a proof of all these properties of the parallel transport unitaries.

An orthonormal basis for $P(k)$ is then obtained by setting
\[ \phi_a(k) := T(k) \, \eu^{-\iu k X} \, \phi_a(0), \quad 1 \le a \le m, \]
for any choice of an orthonormal basis $\set{\phi_1(0), \ldots, \phi_m(0)}$ for the range of $P(0)$. By construction, the vectors $\phi_a(k)$ lie in the range of $P(k)$, are smooth and periodic, and therefore this proves Theorem \ref{thm:main} in the 1D case.

\section{2D case and the first Chern number}

Consider now a 2D family of projections $P(\bk) = P(k_1, k_2)$, $\bk \in \T^2 \simeq \R^2 / 2\pi \Z^2$. By the previous Section, we can assume that we have constructed an orthonormal basis of smooth and $(2\pi\Z)$-periodic Bloch functions $\set{\phi_1(0,k_2), \ldots, \phi_m(0,k_2)}$ for the 1D restriction $P(0,k_2)$, $k_2 \in \T^1$. Our goal is to extend this basis along the $k_1$-direction: once again, we can proceed using the parallel transport unitaries $T(\bk) \equiv T_{k_2}(k_1)$, which are defined again through \eqref{eqn:parallel} by fixing $k_2$ parametrically and letting $k_1$ vary. Notice that the parallel transport unitaries depend periodically on $k_2$, since so does $P(\bk)$ and therefore the generator $K(\bk)$ which drives the $k_1$-dependence in \eqref{eqn:parallel}. We set
\begin{equation} \label{eqn:2D_psi}
\psi_a(k_1,k_2) := T_{k_2}(k_1) \, \phi_a(0,k_2), \quad 1 \le a \le m.
\end{equation}
This time we encounter therefore a $k_2$-dependent holonomy $T_{k_2}(2\pi)$, which by the intertwining property maps the range of $P(0,k_2)$ to the one of $P(2\pi,k_2) = P(0,k_2)$. Define now 
\begin{equation} \label{eqn:matching2D}
\alpha_{ab}(k_2) := \left\langle \phi_b(0,k_2), \, \psi_a(2\pi,k_2) \right\rangle = \left\langle \phi_b(0,k_2), \, T_{k_2}(2\pi) \, \phi_a(0,k_2) \right\rangle, \quad 1 \le a, b \le m.
\end{equation}
We will say that $\alpha(k_2) = \big[ \alpha_{ab}(k_2) \big]_{1 \le a,b \le m}$ defines the family of \emph{matching matrices} associated to the orthonormal basis $\set{\psi_a(k_1,k_2)}_{1 \le a \le m}$: these are the matrix representations of the linear maps $P(2\pi,k_2) \, T_{k_2}(2\pi) \, P(0,k_2)$ from the $m$-dimensional range of $P(0,k_2)$ to the one of $P(2\pi,k_2)=P(0,k_2)$, if both linear spaces are spanned by the orthonormal basis $\set{\phi_a(0,k_2)}_{1 \le a \le m}$. It is therefore easily realized that $\alpha(k_2)$ is a unitary $m \times m$ matrix which depends smoothly and $(2\pi\Z)$-periodically on $k_2$, and we may then view $k_2 \mapsto \alpha(k_2)$ as a map $\alpha \colon \T^1 \simeq \R / 2 \pi\Z \to U(m)$. Roughly speaking, $\alpha(k_2)$ measures the lack of periodicity in the $k_1$-direction for the basis, since
\begin{equation} \label{eqn:Matching}
\psi_b(k_1+2\pi,k_2) = \sum_{a=1}^{m} \psi_a(k_1,k_2) \, \alpha_{ab}(k_2), \quad 1 \le b \le m.
\end{equation}

The next Theorem establishes a link between our original problem of the construction of smooth and periodic Bloch functions and the \emph{homotopy properties} of the matching matrices viewed as unitary-matrix-valued maps on $\T^1$; in turn, these properties are linked with the topology of the Bloch bundle through its Chern number. As the topology of the space of unitary matrices becomes relevant in view of this observation, we begin by setting up some notation: for $\alpha \in U(m)$, we write%
\footnote{It may at times be convenient to notice that $\sigma(\alpha)$ is obtained from $\alpha$ simply by dividing all the entries in the first row by $\det \alpha$.}%
\begin{equation} \label{eqn:U(1)xSU(m)}
\alpha = \delta \, \sigma, \quad \text{where} \quad \delta \equiv \delta(\alpha) := \begin{pmatrix} \det \alpha & 0 \\ 0 & \Id_{m-1} \end{pmatrix} \quad \text{and} \quad \sigma \equiv \sigma(\alpha) := \delta(\alpha)^{-1} \, \alpha \in SU(m).
\end{equation}
The association $\alpha \in U(m) \mapsto (\det \alpha, \sigma) \in U(1) \times SU(m)$ establishes a topological isomorphism $U(m) \simeq U(1) \times SU(m)$. Therefore, in the following we will refer to $\delta(\alpha)$ as the \emph{$U(1)$-part} and to $\sigma(\alpha)$ as the \emph{$SU(m)$-part} of $\alpha \in U(m)$.

\begin{remark}
In the following, we will often freely swap between continuous and smooth deformations (or rather, homotopies). This is without loss of generality, since any continuous homotopy of smooth maps between manifolds is arbitrarily close, in an appropriate topology, to a smooth one, see Ref.\ \onlinecite[Lemma 2.6.3]{mukherjee2015differential}. Such smooth deformations can oftentimes be obtained from continuous ones e.g.\ by convolution with some appropriate kernel, and preserve any symmetry property (like periodicity or ``quasi-periodicity'', see below): smoothing arguments of this sort, for Bloch functions or unitary-matrix-valued functions, can be found in Ref.s\ \onlinecite{cornean2016construction, fiorenza2016construction, fiorenza2016Z2, cornean2017wannier}.
\end{remark}

\begin{theorem} \label{thm:2D}
\begin{enumerate}
 \item \label{item:2D_beta} If the orthonormal basis $\set{\psi_a(k_1,k_2)}_{1 \le a \le m}$ has matching matrices $\alpha(k_2)$, and if $\beta(k_1,k_2)$ is a family of $m \times m$ unitary matrices with $\beta(0,k_2) \equiv \Id$ which is $(2\pi\Z)$-periodic in $k_2$, then the orthonormal basis
\[ \widetilde{\psi}_a(\bk) := \sum_{b=1}^{m} \psi_b(\bk) \, \beta_{ba}(\bk), \quad \bk = (k_1,k_2), \]
has matching matrices
\begin{equation} \label{eqn:beta2D}
\widetilde{\alpha}(k_2) := \beta(k_1,k_2) \, \alpha(k_2) \, \beta(k_1+2\pi,k_2)^{-1}.
\end{equation}
 \item \label{item:2D_alphabeta} Given two maps $\alpha, \widetilde{\alpha} \colon \T^1 \to U(m)$, the following are equivalent.
\begin{enumerate}
 \item \label{item:2D_alphabeta_1} There exist $\beta(\bk) \in U(m)$ as in the previous point such that \eqref{eqn:beta2D} holds.
 \item \label{item:2D_alphabeta_2} The maps $\alpha, \widetilde{\alpha} \colon \T^1 \to U(m)$ are \emph{homotopically equivalent}, that is, they can be continuously deformed one into the other. 
 \item \label{item:2D_alphabeta_3} The \emph{$1$-degrees} of the maps $\alpha, \widetilde{\alpha} \colon \T^1 \to U(m)$ agree:
 \[ \odeg(\alpha) = \odeg(\widetilde{\alpha}) \in \Z, \quad \text{where} \quad \odeg(\alpha) := \frac{1}{2\pi\iu} \int_{\T^1} \Tr_{\C^m} \left( \alpha^{-1} \, \di \alpha \right)\, . \]
\end{enumerate}
 \item \label{item:2D_alphadelta} The family of matrices
 \begin{equation} \label{eqn:2D_alphatilde}
 \alpha(k_2) \quad \text{and} \quad \delta(k_2) := \delta(\alpha(k_2)) = 
 \begin{pmatrix}
 \det \alpha(k_2) & 0 \\ 0 & \Id_{m-1}
 \end{pmatrix}
 \end{equation}
 define homotopically equivalent maps $\alpha, \delta \colon \T^1 \to U(m)$.
 \item \label{item:2Ddeg=c1} It holds that
\begin{equation} \label{eqn:2D_deg_c1} 
\odeg(\alpha) = c_1^{\set{1,2}}(P) \in \Z.
\end{equation}
\end{enumerate}
\end{theorem}

\begin{remark}[$1$-degree and winding number of the determinant] \label{rmk:winding}
It is not difficult to argue (see e.g.\ Ref.\ \onlinecite[Lemma 2.1]{cornean2017wannier}) that, if $\alpha \colon \T^1 \to U(m)$ and $a(k_2) := \det \alpha(k_2) \in U(1)$, then
\[ \Tr_{\C^m}\big(\alpha^{-1} \, \di \alpha \big) = a^{-1} \di a \quad \text{and hence} \quad \odeg(\alpha) = \frac{1}{2\pi\iu} \int_{\T^1} a^{-1} \, \di a. \]
The latter formula computes the \emph{winding number} (or \emph{topological degree}) of the map $a = \det \alpha \colon \T^1 \to U(1)$; in particular, the $1$-degree is integer-valued and is an additive function of its argument $\alpha$, that is, $\odeg(\alpha_1\,\alpha_2) = \odeg(\alpha_1) + \odeg(\alpha_2)$ (see e.g.\ Ref.\ \onlinecite[Section 6.5]{mukherjee2015differential}). The topology of a $U(m)$-valued periodic map is therefore all contained in its $U(1)$-part, which has the determinant as the only non-trivial diagonal entry.
\end{remark} 

Even though Theorem \ref{thm:2D} is proved in Ref.\ \onlinecite[Propositions~5.1, 5.3 and~6.3]{cornean2019parseval}, we will provide a sketch of the proof with some considerations which will be also valuable for the treatment of the 4D situation later.

\begin{proof}[Proof of Theorem \ref{thm:2D}.\ref{item:2D_beta}]
The statement follows from a direct computation, which we leave to the reader.
\end{proof}

\begin{proof}[Proof of Theorem \ref{thm:2D}.\ref{item:2D_alphabeta}]
Let us show first that \eqref{item:2D_alphabeta_1} is equivalent to  \eqref{item:2D_alphabeta_2}. If \eqref{eqn:beta2D} holds, then
\[ \alpha_t(k_2) := \beta(-\pi\, t,k_2) \, \alpha(k_2) \, \beta(\pi \, t,k_2)^{-1}, \quad t \in [0,1], \quad k_2 \in \T^1, \]
defines the desired periodic homotopy between $\alpha$ and $\widetilde{\alpha}$. Conversely, assume that the two family of matrices are continuously deformed one into the other through $\alpha_t$, $t \in [0,1]$, so that in particular $\alpha_{t=0} = \alpha$ and $\alpha_{t=1} = \widetilde{\alpha}$. For $k_1 \in [0, 2\pi]$ and $k_2 \in \T^1$, define
\[ \beta(k_1, k_2) := \alpha_{k_1/2\pi}(k_2)^{-1} \, \alpha(k_2), \]
and extend this definition on other intervals of lenght $2\pi$ in $k_1$ by imposing that
%\[ \beta(k_1+2\pi, k_2) = \widetilde{\alpha}(k_2)^{-1} \, \beta(k_1, k_2) \, \alpha(k_2) \quad \Longleftrightarrow \quad \beta(k_1, k_2) = \widetilde{\alpha}(k_2) \, \beta(k_1+2\pi, k_2) \, \alpha(k_2)^{-1}   \]
%holds. This leads to the desired $\beta$ as in the statement.
\[ \beta(k_1+2\pi,k_2) := \widetilde{\alpha}(k_2)^{-1} \, \beta(k_1,k_2) \,\alpha(k_2) \]
for positive $k_1>0$ and
\[ \beta(k_1,k_2) := \widetilde{\alpha}(k_2) \, \beta(k_1+2\pi,k_2) \, \alpha(k_2)^{-1} \]
for negative $k_1 < 0$.
Notice first that the above defines a family of unitary matrices which is $(2\pi\Z)$-periodic in $k_2$. It remains to show that this definition yields also a continuous function of $k_1$. We have $\beta(0^+,k_2)=\Id$ and $\beta(2\pi^-,k_2)=\widetilde{\alpha}(k_2)^{-1} \, \alpha(k_2)$ by definition. Let $\varepsilon > 0$. If $k_1=-\varepsilon$ is negative but close to zero, we have due to the definition
\[ \beta(-\varepsilon,k_2) = \widetilde{\alpha}(k_2) \, \beta(2\pi-\varepsilon,k_2) \, \alpha(k_2)^{-1} \longrightarrow \widetilde{\alpha}(k_2) \, \beta(2\pi^-,k_2) \, \alpha(k_2)^{-1} = \Id \quad \text{as } \varepsilon \to 0. \]
Hence $\beta$ is continuous at $k_1=0$. At $k_1=2\pi$ we have instead
\[ \beta(2\pi+\varepsilon,k_2) = \widetilde{\alpha}(k_2)^{-1} \, \beta(\varepsilon,k_2) \, \alpha(k_2) \to \widetilde{\alpha}(k_2)^{-1} \, \beta(0^+,k_2) \, \alpha(k_2) = \widetilde{\alpha}(k_2)^{-1} \, \alpha(k_2) \quad \text{as } \varepsilon \to 0 \]
and $\beta$ is also continuous there. A similar argument shows continuity of $k_1 \mapsto \beta(k_1,k_2)$ at every other value of $k_1$ which is an integer multiple of $2\pi$, and therefore on the whole $\R$. This leads to the desired $\beta$ as in the statement.

Next we show that the $1$-degree of a $U(m)$-valued periodic map is an homotopy invariant, and therefore that \eqref{item:2D_alphabeta_2} implies \eqref{item:2D_alphabeta_3}. Without loss of generality, as was mentioned in the previous Remark, we can work with \emph{smooth} homotopies. Therefore, let us pick a deformation $\alpha_s \colon \T^1 \to U(m)$ of some map $\alpha_0$ to some other map $\alpha_1$ which depends smoothly on $s \in [0,1]$. Compute, using the cyclicity of the trace and commuting derivatives,
\begin{align*}
\partial_s \Tr_{\C^m}\big(\alpha_s^{-1} \, \di \alpha_s \big) & = - \Tr_{\C^m}\big(\alpha_s^{-1} \, (\partial_s \alpha_s) \, \alpha_s^{-1} \, \di \alpha_s \big) + \Tr_{\C^m}\big(\alpha_s^{-1} \, \partial_s (\di \alpha_s) \big) = \Tr_{\C^m}\big((\partial_s \alpha_s) \, \di (\alpha_s^{-1} ) \big) + \Tr_{\C^m}\big(\di (\partial_s \alpha_s) \, \alpha_s^{-1}\big) \\
& = \di \, \Tr_{\C^m}\big(\alpha_s^{-1} \, \partial_s \alpha_s \big).
\end{align*}
Consequently
\[ \partial_s \int_{\T^1} \Tr_{\C^m}\big(\alpha_s^{-1} \, \di \alpha_s \big) = \int_{\T^1} \di \, \Tr_{\C^m}\big(\alpha_s^{-1} \, \partial_s \alpha_s \big) = 0 \]
and we can conclude that $\odeg(\alpha_0) = \odeg(\alpha_1)$, as claimed.

Finally we need to show that \eqref{item:2D_alphabeta_3} implies \eqref{item:2D_alphabeta_2}. The general construction of a homotopy between two $U(m)$-valued smooth periodic maps $\alpha$ and $\widetilde{\alpha}$ with the same $1$-degree is presented in Ref.s\ \onlinecite{cornean2017wannier, cornean2019parseval} and uses a \emph{multi-step-logarithm} construction. It is worth pointing out however that the homotopy between a periodic family of unitary matrices and its $U(1)$-part can be explicitly produced also by a different method, described in Ref.\ \onlinecite{gontier2019numerical} as the \emph{column interpolation method}. The basic observation is that each column of the matrix $\alpha(k_2)\in U(m)$ defines a vector in the unit sphere $S^{2m-1} \subset \C^m$. Let us focus on the last column for concreteness. Provided $2m-1 > 1 = \dim \T^1$, that is, $m > 1$, this smooth periodic family of column vectors cannot cover the whole sphere by Sard's lemma, and can therefore be deformed to a constant loop for example by contracting its stereographic projection from a generic point (in the measure-theoretic sense). All the other columns of $\alpha(k_2)$ can be deformed along by retaining the unitarity constraint, which imposes orthogonality with respect to the last column, by using a parallel-trasport argument: this allows to deform the whole matrix $\alpha(k_2)$ \emph{through unitary matrices} to one which has a constant last column, say fixed equal to the last vector in the standard basis of $\C^m$. The resulting periodic family of unitary matrices are therefore in block-diagonal form: an iteration of this deformation argument applied successively to the ``last'' column vectors of the upper-left blocks allows to produce the $(m-1) \times (m-1)$ identity block claimed in the statement and therefore continuously deform $\alpha$ to a diagonal matrix $\alpha_1$ with a single (possibly) non-constant entry $a(k_2) = \det \alpha_1(k_2) \in U(1)$. 
 
In view of Remark \ref{rmk:winding}, we must have that $a=\det \alpha_1 \colon \T^1 \to U(1)$ defines the same homotopy class as $\det \alpha \colon \T^1 \to U(1)$, which in turns allows to continuously deform $\alpha_1$ into $\widetilde{\alpha}$ as in \eqref{eqn:2D_alphatilde}. Explicitly, this can be done as follows: consider the map $f := a^{-1} \, \det \alpha \colon \T^1 \to U(1)$, and observe that this has vanishing winding number due to the additivity of the latter. Observe that, if we can show that $f$ can be continuosly deformed to the constant map $1$ through continuous periodic maps $f_t$, $t \in [0,1]$, then $a$ can be continuously deformed to $\det \alpha$ through $a \cdot f_t$. We will now construct an homotopy $f_t \colon \T^1 \to U(1)$, $t \in [0,1]$.

By uniform continuity of $f \colon \T^1 \simeq [0,2\pi] \to U(1)$, there exists $N \in \N$ such that
\[ |f(k_2) - f(k_2')| < 2 \quad \text{as long as} \quad |k_2-k_2'| \le \frac{2\pi}{N}. \]
The above inequality implies in particular that $f(k_2)\,f(0)^{-1}$ cannot equal $-1$ for all $k_2 \in [0, 2\pi/N]$, and therefore, by choosing a branch cut for the logarithm function on the negative real semi-axis in $\C$, we can write
\[ f(k_2)\,f(0)^{-1} = \eu^{2 \pi \iu \, \theta_0(k_2)} \quad \text{for} \quad k_2 \in \left[0,\frac{2\pi}{N}\right], \quad \text{with} \quad \theta_0(0)=0. \]
Arguing similarly, we can write
\[ f(k_2)\,f\left(\frac{2\pi j}{N}\right)^{-1} = \eu^{2 \pi \iu \, \theta_j(k_2)} \quad \text{for} \quad k_2 \in \left[\frac{2\pi j}{N},\frac{2\pi(j+1)}{N}\right], \quad \text{with} \quad \theta_j\left(\frac{2\pi j}{N}\right)=0, \quad j \in \set{0, \ldots, N-1}. \]
The above normalization for $\theta_j(2\pi j/N)$ is chosen so that the following provides a \emph{continuous} choice of the argument for $f \colon \T^1 \to U(1)$:
\[ f(k_2) = \eu^{2\pi \iu \, \theta(k_2)} \, f(0), \quad \text{where} \quad \theta(k_2) := \sum_{j=0}^{N-1} \left[ \theta_j(k_2) \, \Id\left(\frac{2\pi j}{N} \le k_2 < \frac{2\pi (j+1)}{N}\right) + \sum_{l = 0}^{j-1} \theta_l \left( \frac{2\pi (l+1)}{N} \right) \right]. \]
Notice now that
\begin{equation} \label{eqn:periodicargument}
\theta(2\pi) - \theta(0) = \int_{\T^1} \di \theta = \frac{1}{2\pi\iu} \int_{\T^1} f^{-1} \, \di f = 0,
\end{equation}
as $f$ is supposed to have a vanishing winding number. We conclude that the argument $\theta$ can also be choosen to be periodic. Therefore, $f_t(k_2) := \eu^{2\pi\iu \, (1-t) \, \theta(k_2)} \, f(0)^{1-t}$ defines a continuous deformation $f_t \colon \T^1 \to U(1)$ of $f$ to the map constantly equal to~$1$, as desired.
\end{proof}

\begin{proof}[Proof of Theorem \ref{thm:2D}.\ref{item:2D_alphadelta}]
It follows from Theorem \ref{thm:2D}.\ref{item:2D_alphabeta}, as the two maps $\alpha$ and $\delta$ clearly have the same determinant, and therefore the same $1$-degree.
\end{proof}

In preparation for the proof of Theorem \ref{thm:2D}.\ref{item:2Ddeg=c1}, we make here some general remarks regarding a $d$-dimensional family of projections $P(\bk)$, $\bk \in \T^d$, and the associated Berry curvature $F$ defined in \eqref{eqn:Berry_curv}. For the sake of a self-contained presentation, the proof of these statements is deferred to Appendix \ref{app:CS}, see also the references therein.

Given an orthonormal set $\Phi:=\set{\phi_1(\bk), \ldots, \phi_m(\bk)}$ of smooth Bloch functions, introduce the \emph{Berry connection}, which is the matrix-valued $1$-form $A = (A_{ab})_{1 \le a,b \le m}$ given by
\begin{equation} \label{eqn:BerryConnection}
A_{ab} \equiv A^{(\Phi)}_{ab} = \sum_{1 \le \mu \le d} A_{\mu}(\bk)_{ab} \, \di k_\mu, \quad \text{with} \quad A_{\mu}(\bk)_{ab} := \frac{1}{2\pi\iu} \left\langle \phi_a(\bk) , \, \partial_{k_\mu} \phi_b(\bk) \right\rangle, \quad 1 \le a, b \le m.
\end{equation}
The Berry curvature $F$ is the curvature $2$-form associated to this connection, in the sense that
\[ F = \di A + 2 \pi \iu \, A \wedge A \]
which spells out to
\begin{equation} \label{F=dA-[A,A]} 
F_{\mu\nu}(\bk) = \partial_{\mu} A_{\nu}(\bk) - \partial_{\nu} A_{\mu}(\bk) + 2 \pi \iu \, \big[ A_{\mu}(\bk), A_{\nu}(\bk) \big], \quad 1 \le \mu < \nu \le d\,.
\end{equation} 
In particular, it holds that
\[ c_1(P) = \Tr(F) = \di \, \Tr_{\C^m}(A). \]
The equality above should \emph{not} be interpreted as the first Chern form being exact, as it holds only ``locally'', that is, as long as a smooth orthonormal basis for the range of $P(\bk)$ is defined; indeed, a closed form (such as the first Chern form) is in general only locally exact, by the Poincar\'{e} lemma. As we saw, in general a smooth choice of orthonormal Bloch functions cannot be also made to depend periodically on $\bk$, and therefore $A$ is not a globally-defined $1$-form on the torus $\T^d$. Even more so, the expression provided above for the Berry connection depends non-trivially on the choice of the orthonormal set $\Phi$, or, as one says, it's \emph{gauge-dependent}. Indeed, if $\Psi = \set{\psi_1(\bk), \ldots, \psi_m(\bk)}$ is a different choice of smooth orthonormal Bloch functions, then necessarily
\[ \psi_{a}(\bk) = \sum_{b=1}^{m} \phi_b(\bk) \, \gamma_{ba}(\bk) \]
for some unitary matrix $\gamma(\bk) \in U(m)$, and
\begin{equation} \label{eqn:A_gauge}
A^{(\Psi)} = \gamma^{-1} \, A^{(\Phi)} \, \gamma + \frac{1}{2\pi\iu} \, \gamma^{-1} \, \di \gamma.
\end{equation}
Combining \eqref{F=dA-[A,A]} and \eqref{eqn:A_gauge} one can compute that the Berry curvature is instead gauge-covariant, that is,
\[ F^{(\Psi)} = \gamma^{-1} \, F^{(\Phi)} \, \gamma \]
and therefore traces of its exterior powers, which enter in the definition of the Chern forms, are gauge-independent: this corroborates the fact that the Chern forms themselves only depend on the spaces spanned by the family of projections (and not on the choice of an orthonormal basis for them).

\begin{proof}[Proof of Theorem \ref{thm:2D}.\ref{item:2Ddeg=c1}]
Consider the Berry connection $A^{(\Psi)}$ associated to the Bloch functions defined in \eqref{eqn:2D_psi}. Those are not periodic with respect to $k_1$ but are nonetheless smooth on the whole $\T^2 \simeq [0,2\pi]^2$. Using Stokes' theorem, one can then compute
 \[ c_1^{\set{1,2}}(P) = \int_{\T^2} \Tr(F) = \int_{[0,2\pi]^2} \di \Tr_{\C^m}(A) = \int_{\partial [0,2\pi]^2} \Tr_{\C^m}(A) = \int_{0}^{2\pi} \Tr_{\C^m}\big(A_{2}(2\pi,k_2)\big) \, \di k_2 - \int_{0}^{2\pi} \Tr_{\C^m}\big(A_{2}(0,k_2)\big) \, \di k_2 \]
 (the other two sides of $\partial [0,2\pi]^2$, where $k_2=0$ or $k_2=2\pi$, do not contribute to the last equality, as the integrand $\Tr(A)$ is periodic in $k_2$ and the sides have opposite orientations). Let us now recall that $A_{2}(0,k_2)$, respectively $A_{2}(2\pi,k_2)$, is computed from the vectors $\Psi(0,\cdot) := \set{\psi_a(0,k_2) = \phi_a(0,k_2)}_{1\le a \le m}$, respectively from the vectors $\Psi(2\pi,\cdot) := \set{\psi_a(2\pi,k_2)}_{1 \le a \le m}$ which are related to the vectors $\phi_b(0,k_2)$ by the matching matrices $\alpha(k_2)$. Therefore, using \eqref{eqn:Matching} and \eqref{eqn:A_gauge} we conclude
 \[ c_1^{\set{1,2}}(P) = \int_{\{k_1=2\pi\}} \Tr_{\C^m}\left(A^{(\Psi(2\pi,\cdot))}\right)  - \int_{\{k_1=0\}} \Tr_{\C^m}\left(A^{(\Psi(0,\cdot))}\right) = \frac{1}{2\pi\iu} \int_{0}^{2\pi} \Tr_{\C^m}\left( \alpha(k_2)^{-1} \, \partial_{k_2} \alpha(k_2) \right) \, \di k_2 = \odeg(\alpha) \]
as claimed.
\end{proof}

\begin{remark}[Normal form for the 2D matching matrices] \label{rmk:normalalpha2D}
In view of Theorem \ref{thm:2D}.\ref{item:2D_alphabeta}, the matrix-valued map $\alpha$ defined by the matching matrices \eqref{eqn:matching2D} is homotopic to the map $\alpha\sub{2D} \colon \T^1 \to U(m)$ given by
\begin{equation} \label{eqn:alpha2D}
\alpha\sub{2D}(k_2) := \begin{pmatrix} 
\eu^{\iu \, n_2 \, k_2} & 0 \\ 0 & \Id_{m-1}
\end{pmatrix}, \quad \text{where} \quad n_2 := c_1^{\set{1,2}}(P) \in \Z.
\end{equation}
\end{remark}

We combine now the various statements in Theorem \ref{thm:2D} to recover the results of Ref.\ \onlinecite{cornean2019parseval}, proving Theorem \ref{thm:main} in $d=2$.

\begin{proof}[Proof of Theorem \ref{thm:main}, $d=2$]
As for point \ref{item:orthonormal} in the statement of Theorem \ref{thm:main}, Theorem \ref{thm:2D} and the previous Remark ensure that the smooth Bloch functions from \eqref{eqn:2D_psi} can be transformed, through a unitary matrix $\beta(\bk) \in U(m)$, to Bloch functions $\set{\widetilde{\psi}_a(\bk)}_{1 \le a \le m}$ which are orthonormal, smooth in $\bk$, $(2\pi\Z)$-periodic in $k_2$, and such that
\[ \widetilde{\psi}_1(k_1+2\pi,k_2) = \eu^{\iu \, c_1^{\set{1,2}}(P) \, k_2} \cdot \widetilde{\psi}_1(k_1,k_2) \quad \text{while} \quad \widetilde{\psi}_a(k_1+2\pi,k_2) = \widetilde{\psi}_a(k_1,k_2) \quad \text{for } 2 \le a \le m, \]
that is, the vectors $\set{\widetilde{\psi}_a(\bk)}_{2 \le a \le m}$ are also $(2\pi\Z)$-periodic in $k_1$ while the vector $\widetilde{\psi}_1(\bk)$ picks up a phase $\eu^{\iu \, c_1^{\set{1,2}}(P) \, k_2}$ when going through a loop in the direction $k_1$. If we further assume that $c_1^{\set{1,2}}(P) = 0$, that is, if the Bloch bundle is Chern-trivial, then all the vectors are periodic in $k_1$ as well, and we are done.

Let us now come to point \ref{item:Parseval} in the statement of Theorem \ref{thm:main}. In the general case $c_1^{\set{1,2}}(P) \ne 0$, all the topology has been ``squeezed'' in a rank-1 subprojection $P_1(\bk) = \ket{\widetilde{\psi}_1(\bk)} \bra{\widetilde{\psi}_1(\bk)}$ of $P(\bk)$. Notice that, although $\widetilde{\psi}_1(\bk)$ is not periodic in $k_1$, the projection $P_1(\bk)$ is. Moreover by construction $c_1^{\set{1,2}}(P_1) = c_1^{\set{1,2}}(P)$. A technique dubbed \emph{space-doubling trick} in Ref.\ \onlinecite{cornean2019parseval} allows now to promote the non-periodic vector $\widetilde{\psi}_1(\bk)$ to two smooth and periodic vectors $\Psi_1(\bk)$ and $\Psi_2(\bk)$ which are not orthonormal but still span the same subspace, i.e.\ the 1-dimensional range of $P_1(\bk)$. The procedure requires to complement the projection $P_1$ with another rank-1 projection $Q_1$ of opposite Chern number: this can be done by taking for example \cite{panati2007triviality} $Q_1(\bk) = \overline{P_1(-\bk)}$. The rank-2 projection $P_1 \oplus Q_1$ on $\Hi \oplus \Hi$ is then Chern-trivial and can be therefore spanned by two smooth and periodic ``two-legged'' Bloch functions in $\Hi \oplus \Hi$. The projections of the latter to the first leg in the direct sum give the desired spanning vectors for the range of $P_1$, which together with the orthonormal vectors $\set{\widetilde{\psi}_a(\bk)}_{2 \le a \le m}$ constitute the desired Parseval frame of smooth and periodic Bloch functions for $P(\bk)$. This completes the proof of Theorem \ref{thm:main} in $d=2$.
\end{proof}

\section{3D case and the weak invariants}

We now move to the 3-dimensional (3D) case and consider a family of rank-$m$ projections $P(\bk)$ depending smoothly and periodically on $\bk \in \T^3 \simeq \R^3 / 2 \pi \Z^3$. In view of Theorem \ref{thm:2D}, we can assume that the restriction $P(0, k_2, k_3)$ is spanned by an orthonormal set of smooth Bloch functions $\set{\phi_a(0,k_2,k_3)}_{1 \le a \le m}$ such that
\begin{equation} \label{eqn:3D_phi}
\begin{gathered}
\phi_a(0, k_2, k_3 + 2\pi) = \phi_a(0, k_2, k_3) \quad \text{for all } a \in \set{1, \ldots, m}, \\
\phi_a(0, k_2 + 2\pi, k_3) = \phi_a(0, k_2, k_3) \quad \text{for all } a \in \set{2, \ldots, m}, \quad \text{while} \quad \phi_1(0, k_2 + 2\pi, k_3) = \eu^{\iu \, c_1^{\set{2,3}}(P) \, k_3} \, \phi_1(0, k_2, k_3).
\end{gathered}
\end{equation}
This corresponds to the set of Bloch functions having $\alpha\sub{2D}(k_3)$ in \eqref{eqn:alpha2D} as their matching matrices. In analogy with what was done in \eqref{eqn:2D_psi} in the previous Section, we extend the definition of these Bloch vectors by means of parallel transport in the $k_1$-direction, and set
\begin{equation} \label{eqn:3D_psi}
\psi_a(k_1,k_2,k_3) := T_{(k_2,k_3)}(k_1) \, \phi_a(0,k_2,k_3), \quad 1 \le a \le m,
\end{equation}
where the parallel transport unitaries are defined by the Cauchy problem \eqref{eqn:parallel} and depend parametrically (in a smooth and periodic way) on $(k_2, k_3) \in \T^2$. Once again we can consider the matrix representatives of the holonomy unitaries $T_{(k_2,k_3)}(2\pi)$ in the basis selected above for the range of $P(0,k_2,k_3) = P(2\pi,k_2,k_3)$: this leads to the definition of the matriching matrices
\begin{equation} \label{eqn:3Dmatching}
\alpha_{ab}(k_2,k_3) := \left\langle \phi_b(0,k_2,k_3), \, T_{(k_2,k_3)}(2\pi) \, \phi_a(0,k_2,k_3) \right\rangle, \quad 1 \le a, b \le m.
\end{equation}

The matrices $\alpha(k_2,k_3)$ are unitary, they depend $(2\pi\Z)$-periodically on $k_3$, but due to \eqref{eqn:3D_phi} they satisfy the following condition when one goes along a full loop in the $k_2$-direction%
\footnote{Explicitly, on the right-hand side of the equality below, the first row of the matrix $\alpha(k_2,k_3)$ is multiplied by the phase $\eu^{\iu n_3 k_3}$, while the first column is multiplied by the opposite phase; here, $n_3 = c_1^{\set{2,3}}(P)$.}:
\begin{equation} \label{eqn:alpha2D-periodicity}
\alpha(k_2+2\pi,k_3) = \alpha\sub{2D}(k_3)^{-1} \, \alpha(k_2,k_3) \, \alpha\sub{2D}(k_3)\,.
\end{equation}
The above relation is called \emph{$\alpha\sub{2D}$-periodicity} in the $k_2$-direction in Ref.\ \onlinecite[Definition~6.1]{cornean2019parseval}. Notice how, if $c_1^{\set{2,3}}(P)=0$, then $\alpha\sub{2D} \equiv \Id_m$ and $\alpha\sub{2D}$-periodicity reduces to mere periodicity in the $k_2$-direction.

The following result generalizes Theorem~\ref{thm:2D} to the 3D case; it already appeared as Ref.\ \onlinecite[Proposition~6.3]{cornean2019parseval}.

\begin{theorem} \label{thm:3D}
\begin{enumerate}
 \item \label{item:3D_beta} If the orthonormal basis $\set{\psi_a(k_1,k_2,k_3)}_{1 \le a \le m}$ has matching matrices $\alpha(k_2,k_3)$, and if $\beta(k_1,k_2,k_3)$ is a family of $m \times m$ unitary matrices with $\beta(0,k_2,k_3) \equiv \Id$ which is $\alpha\sub{2D}$-periodic in $k_2$ and $(2\pi\Z)$-periodic in $k_3$, then the orthonormal basis
\[ \phi_a(\bk) := \sum_{b=1}^{m} \psi_b(\bk) \, \beta_{ba}(\bk), \quad \bk = (k_1,k_2,k_3), \]
has matching matrices
\begin{equation} \label{eqn:beta3D}
\widetilde{\alpha}(k_2,k_3) := \beta(k_1,k_2,k_3) \, \alpha(k_2,k_3) \, \beta(k_1+2\pi,k_2,k_3)^{-1}.
\end{equation}
 \item Given two $U(m)$-valued maps $\alpha, \widetilde{\alpha}$ which are $\alpha\sub{2D}$-periodic in $k_2$ and $(2\pi\Z)$-periodic in $k_3$, the following are equivalent.
\begin{enumerate}
 \item \label{item:3D_alphabeta_1} There exist $\beta(\bk) \in U(m)$ as in the previous point such that \eqref{eqn:beta3D} holds.
 \item \label{item:3D_alphabeta_2} The maps $\alpha, \widetilde{\alpha}$ are \emph{$\alpha\sub{2D}$-periodically homotopically equivalent}, that is, they can be continuously deformed one into the other through maps which are $\alpha\sub{2D}$-periodic in $k_2$ and $(2\pi\Z)$-periodic in $k_3$.
 \item \label{item:3D_alphabeta_3} The two \emph{$1$-degrees} of the maps $\alpha, \widetilde{\alpha}$ agree: for $j \in \set{2,3}$
 \[ \odeg_j(\alpha) = \odeg_j(\widetilde{\alpha}) \in \Z  \quad \text{with} \quad \odeg_j(\alpha) := \odeg \left( \alpha \Big|_{\T^1_{\set{j}}} \right) \]
(recall the definition of $\T^n_I \subset \T^d$ given below \eqref{eqn:ChernNumbers}).
\end{enumerate}
 \item \label{item:3D_alphabeta} The family of matrices
 \begin{equation} \label{eqn:3D_alphatilde}
 \alpha(k_2,k_3) \quad \text{and} \quad \delta(k_2,k_3) := \delta(\alpha(k_2,k_3)) = 
 \begin{pmatrix}
 \det \alpha(k_2,k_3) & 0 \\ 0 & \Id_{m-1}
 \end{pmatrix}
 \end{equation}
 define $\alpha\sub{2D}$-periodically homotopically equivalent $U(m)$-valued maps $\alpha, \widetilde{\alpha}$.
 \item \label{item:3Ddeg=c1} For $j \in \set{2,3}$, it holds that
\begin{equation} \label{eqn:3D_deg_c1} 
\odeg_j(\alpha) = c_1^{\set{1,j}}(P) \in \Z.
\end{equation}
\end{enumerate}
\end{theorem}
\begin{proof}
The statement of point \ref{item:3D_beta} follows once again from a direct computation.

The equivalence between \eqref{item:3D_alphabeta_1} and \eqref{item:3D_alphabeta_2} can be argued exactly as in the proof of Theorem \ref{thm:2D}; this time, one just has to notice that all matrices involved are $\alpha\sub{2D}$-periodic rather than $(2\pi\Z)$-periodic in $k_2$. 

If \eqref{item:3D_alphabeta_2} holds, then notice that to compute the $1$-degrees of the maps $\alpha$ and $\widetilde{\alpha}$ one just needs to compute winding numbers of their determinants, according to Remark \ref{rmk:winding}. Notice also that, even if $\alpha$ and $\widetilde{\alpha}$ are in general \emph{not} periodic in $k_2$ -- rather they satisfy \eqref{eqn:alpha2D-periodicity} -- their determinants define periodic functions of both $k_2$ and $k_3$. Therefore, any ($\alpha\sub{2D}$-periodic) homotopy of the maps $\alpha$ and $\widetilde{\alpha}$ induces a regular homotopy of the determinant of their restrictions to $\T^1_{\set{j}}$, $j \in \set{2,3}$, as periodic maps. In view of Theorem \ref{thm:2D}, this implies the equality between the $1$-degrees of $\alpha$ and $\widetilde{\alpha}$.

To complete the proof of point \ref{item:3D_alphabeta}, one would need to show, conversely, how to construct an homotopy between maps which have  equal $1$-degrees. For later reference, we factor this construction in the proof of Proposition \ref{prop:mu-nu-periodic} below, where we will even impose a ``quasi''-periodicity in both directions (and not just in the direction of $k_2$). This also proves point \ref{item:3D_alphabeta} of the statement, as the matrices $\alpha$ and $\delta$ in \eqref{eqn:3D_alphatilde} clearly share the same determinant and therefore have equal $1$-degrees.

Finally, point \ref{item:3Ddeg=c1} of the statement is a direct consequence of Theorem \ref{thm:2D}.\ref{item:2Ddeg=c1}.
\end{proof}

\begin{proposition} \label{prop:mu-nu-periodic}
Let $\mu, \nu \colon \T^1 \to U(m)$ be smooth periodic maps of the form
\[ \mu(k_2) = \begin{pmatrix} \det \mu(k_2) & 0 \\ 0 & \Id_{m-1} \end{pmatrix}, \quad \nu(k_3) = \begin{pmatrix} \det \nu(k_3) & 0 \\ 0 & \Id_{m-1} \end{pmatrix}. \]
Assume that $\alpha(k_2,k_3) \in U(m)$ is \emph{$(\mu,\nu)$-periodic}, that is,
\begin{equation} \label{eqn:mu-nu-periodic}
\alpha(k_2 + 2\pi, k_3) = \nu(k_3)^{-1} \, \alpha(k_2,k_3) \, \nu(k_3), \quad \alpha(k_2, k_3 + 2\pi) = \mu(k_2)^{-1} \, \alpha(k_2,k_3) \, \mu(k_2).
\end{equation}
(Notice that the two conditions above are compatible because the diagonal matrices $\mu$ and $\nu$ commute.) Then
\begin{equation} \label{eqn:alpha-tildealpha-mu-nu}
\alpha(k_2,k_3) \quad \text{and} \quad \delta(k_2,k_3) = \delta(\alpha(k_2,k_3)) := \begin{pmatrix} \det \alpha(k_2,k_3) & 0 \\ 0 & \Id_{m-1} \end{pmatrix}
\end{equation}
are \emph{$(\mu,\nu)$-periodically homotopically equivalent}, that is, they can be continuously defomed one into the other via $(\mu,\nu)$-periodic $U(m)$-valued maps.
\end{proposition}

In the terminology established in the above statement, the matching matrices defined in \eqref{eqn:3Dmatching} are $(\Id,\alpha\sub{2D})$-periodic.

\begin{proof}[Proof of Proposition \ref{prop:mu-nu-periodic}]
While the multi-step-logarithm construction presented in Ref.s\ \onlinecite{cornean2017wannier, cornean2017construction} has been already generalized to the $\alpha\sub{2D}$-periodic setting \cite{cornean2019parseval}, the column interpolation method from Ref.\ \onlinecite{gontier2019numerical} briefly described in the proof of Theorem \ref{thm:2D}.\ref{item:2D_alphabeta} applies to \emph{$(2\pi\Z^2)$-periodic} maps $\alpha \colon \T^2 \to U(m)$, which can then be brought to a diagonal form $\delta$ as in \eqref{eqn:3D_alphatilde}, but not to $\alpha\sub{2D}$-periodic maps. We sketch here how to adapt the proof presented in Ref.\ \onlinecite{gontier2019numerical} to cover this more general case.

Let us restrict our attention to $(k_2,k_3) \in [0,2\pi]^2$. The unitary matrices $\alpha(k_2,k_3)$ satisfy
\[ \alpha(2\pi,k_3) = \nu(k_3)^{-1} \, \alpha(0,k_3) \, \nu(k_3), \quad \text{with} \quad \alpha(0,2\pi) = \mu(0)^{-1} \, \alpha(0,0) \, \mu(0). \]
Write $\mu(0) =: \eu^{2 \pi \iu M}$, where $M = M^*$ is a self-adjoint $m \times m$ matrix which is diagonal and has (possibly) only the first diagonal entry different from $0$. Then
\[ \gamma(k_3) := \eu^{\iu \, k_3 \, M} \, \alpha(0,k_3) \, \eu^{-\iu \, k_3 \, M} \]
defines a \emph{$(2\pi\Z)$-periodic} map with values in $U(m)$. As such, it can then be deformed to a diagonal matrix as in \eqref{eqn:2D_alphatilde}: there exists $\gamma_s \colon \T^1 \to U(m)$, $s \in [0,1]$, such that
\[ \gamma_{s=0}(k_3) = \gamma(k_3) \quad \text{and} \quad \gamma_{s=1}(k_3) = \begin{pmatrix} \det \gamma(k_3) & 0 \\ 0 & \Id_{m-1} \end{pmatrix} = \begin{pmatrix} \det \alpha(0,k_3) & 0 \\ 0 & \Id_{m-1} \end{pmatrix}. \] 
Letting now $\alpha_s(0,k_3) := \eu^{-\iu \, k_3 \, M} \, \gamma_s(k_3) \, \eu^{\iu \, k_3 \, M}$, for $k_3 \in [0,2\pi]$ and $s \in [0,1]$, we obtain a smooth deformation which for fixed $s$ satisfies the same quasi-periodicity of $\alpha(0,\cdot)$ and which interpolates
\[ \alpha_{s=0}(0,k_3) = \alpha(0,k_3) \quad \text{and} \quad \alpha_{s=1}(0,k_3) = \widetilde{\alpha}(0,k_3) := \begin{pmatrix} \det \alpha(0,k_3) & 0 \\ 0 & \Id_{m-1} \end{pmatrix}. \]
In view of \eqref{eqn:alpha-tildealpha-mu-nu}, it is clear that 
\[ \alpha_s(2\pi,k_3) := \nu(k_3)^{-1} \, \alpha_s(0,k_3) \, \nu(k_3), \quad k_3 \in [0,2\pi], \; s \in [0,1], \]
defines moreover a smooth map which interpolates between
\[ \alpha_{s=0}(2\pi,k_3) = \alpha(2\pi,k_3) \quad \text{and} \quad \alpha_{s=1}(2\pi,k_3) = \widetilde{\alpha}(2\pi,k_3) := \begin{pmatrix} \det \alpha(2\pi,k_3) & 0 \\ 0 & \Id_{m-1} \end{pmatrix}. \]

Changing the roles of $\mu$ and $\nu$, we can similarly find on the horizontal sides of the boundary of $[0,2\pi]^2$ that $\alpha(k_2,0)$ is smoothly interpolated via $\alpha_s(k_2,0)$ to a diagonal matrix $\widetilde{\alpha}(k_2,0)$ with only one (possibly) non-constant entry. This interpolation can be chosen to coincide, at $k_2=0$, with the one already considered by restricting $\alpha_s(0,k_3)$ to $k_3=0$. Imposing the appropriate quasi-periodicity, we can then find a corresponding interpolation between $\alpha(k_2,2\pi)$ and the diagonal $\widetilde{\alpha}(k_2,2\pi)$. We have therefore constructed a homotopy of the restriction of $\alpha$ to the boundary of $[0,2\pi]^2$ to matrices of diagonal form, which is continuous on the whole boundary, and moreover this homotopy is by construction compatible with $(\mu,\nu)$-periodicity. 

\begin{figure}[htb]
\centering
\begin{tikzpicture}[>=stealth]
\filldraw [lightgray] (-1.5,-1.5) rectangle (1.5,1.5);
\filldraw [white] (-1,-1) rectangle (1,1);
\draw (-1,-1) rectangle (1,1);
\draw (-1.5,-1.5) rectangle (1.5,1.5);
\draw (-1.5,-1.5) -- (-1,-1)
      (-1.5,1.5) -- (-1,1)
      (1.5,-1.5) -- (1,-1)
      (1.5,1.5) -- (1,1);
\foreach \x in {-1,0,1}{
\draw [->] (1.1,{1.1*tan(30*\x)}) -- +({.3*cos(30*\x)},{.3*sin(30*\x)});
}
\foreach \x in {2,3,4}{
\draw [->] ({1.1*cos(30*\x)/sin(30*\x)},1.1) -- +({.3*cos(30*\x)},{.3*sin(30*\x)});
}
\foreach \x in {5,6,7}{
\draw [->] (-1.1,{-1.1*tan(30*\x)}) -- +({.3*cos(30*\x)},{.3*sin(30*\x)});
}
\foreach \x in {8,9,10}{
\draw [->] ({-1.1*cos(30*\x)/sin(30*\x)},-1.1) -- +({.3*cos(30*\x)},{.3*sin(30*\x)});
}\end{tikzpicture}
\caption{``Thickened'' box $[-1,2\pi+1]^2$, containing $[0,2\pi]^2$ (here in white)}
\label{fig:thicc}
\end{figure}
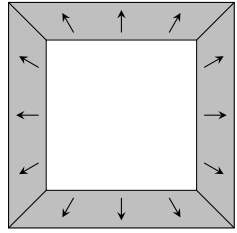

Consider now the ``thickened'' box $[-1,2\pi+1]^2$. Parametrize the radial direction in $[-1,2\pi+1]^2 \setminus [0,2\pi]^2$ (the grey area in Figure \ref{fig:thicc}) via $s \in [0,1]$. Extend now the definition of $\alpha$ from $[0,2\pi]^2$ to the ``thickened'' box as follows: for example, in the left-most quadrant of the grey area, each point is parametrized by $k_3 \in [0,2\pi]$ and the radial coordinate $s \in [0,1]$ as explained previously; at this point, the matrix $\alpha_s(0,k_3)$ should be used to extend the definition of $\alpha$. The different homotopies constructed previously should be used for the different quadrants of the grey area $[-1,2\pi+1]^2 \setminus [0,2\pi]^2$.

We now define unitary matrices $\widehat{\alpha}_t(k_2,k_3)$ as follows: consider the restriction of the above extension of $\alpha$ to the box $[-t,2\pi+t]^2$, and rescale this larger box to $[0,2\pi]^2$. By construction, the restrictions to the boundary of $[0,2\pi]^2$ of these unitary-valued maps are compatible with $(\mu,\nu)$-periodicity. At $t=0$, this family of matrices coincides with the original $\alpha$, and at $t=1$ the matrices $\widehat{\alpha}_{t=1}(k_2,k_3)$ for $(k_2,k_3) \in \partial [0,2\pi]^2$ are diagonal. As such, for $\widehat{\alpha}_{t=1}(k_2,k_3)$ the notion of $(\mu,\nu)$-periodicity coincides with mere $(2\pi\Z^2)$-periodicity, as $\alpha\sub{2D}$ is itself diagonal. We consider now the \emph{$(2\pi\Z^2)$-periodic} extension of $\widehat{\alpha}_{t=1}$ to the whole~$\R^2$, which we regard then as a map $\widehat{\alpha}_{t=1} \colon \T^2 \to U(m)$. The results from Ref.\ \onlinecite{gontier2019numerical} apply to this maps, and it can be deformed further to diagonal matrices with only one non-trivial entry. By combining this deformation with a ``thickening'' argument of the type described above, we can assume that the deformation is ``relative to the boundary'', that is, that the restriction of the unitary-valued map to the boundary of $[0,2\pi]^2$ retains the diagonal form (which is extended to the interior of the cell by means of the column interpolation argument).

Call $\alpha_t$ the family of matrices on $[0,2\pi]^2$ which emerges as the restriction to this cell of the one obtained from the above construction (possibly rescaling $t \in [0,1]$). We impose now \eqref{eqn:mu-nu-periodic} to extend this definition to $(k_2,k_3) \in \R^2$ in a $(\mu,\nu)$-periodic way. We end up with a family of unitary matrices $\alpha_t(k_2,k_3)$ which satisfies the required $(\mu,\nu)$-periodicity conditions, depends smoothly on $t \in [0,1]$, and interpolates bewteen the original $\alpha$ and the diagonal matrix 
\[ \widetilde{\alpha}(k_2,k_3) = \begin{pmatrix} \det \widetilde{\alpha}(k_2,k_3) & 0 \\ 0 & \Id_{m-1} \end{pmatrix}. \]
To conclude the proof, we need to show that $a := \det \alpha$ and $\widetilde{a} := \det \widetilde{\alpha}$ are homotopic one to the other as maps $\T^2 \to U(1)$. As in the proof of Theorem \ref{thm:2D}.\ref{item:2D_alphabeta}, it suffices to show that $f := a^{-1} \, \widetilde{a} \colon \T^2 \to U(1)$ is homotopic to the constant map. We observe first of all that the restrictions of $f$ to the sub-tori $\T^1_{\set{j}} \subset \T^2$, $j \in \set{2,3}$, have vanishing winding numbers: we will now argue that this is the key condition to guarantee that $f$ can be continuously deformed to the map constantly equal to $1$.

Let us consider $f(k_2,k_3)$ as a periodic and continuous function of $k_2$ with values in continuous and periodic functions of $k_3$. Endowing the latter space with the $\sup$-norm, by uniform continuity we see that there exists $N \in \N$ such that
\[ \sup_{k_3 \in \T^1} \left| f(k_2, k_3) - f(k_2', k_3) \right| < 2 \quad \text{as long as} \quad |k_2 - k_2'| \le \frac{2\pi}{N}. \]
But then $f(k_2,k_3) \, f(0,k_3)^{-1}$ never assumes the value $-1$, and one can define a continuous argument
\[ f(k_2,k_3) \, f(0,k_3)^{-1} = \eu^{2 \pi \iu \, \theta_0(k_2,k_3)}, \quad k_2 \in \left[0, \frac{2\pi}{N}\right], \; k_3 \in \T^1. \]
Notice in particular that $k_3 \mapsto f(k_2,k_3)$ and $k_3 \mapsto f(0,k_3)$ have the same winding number (equal to zero, but this is not needed at the moment), since they can be deformed into one another continuously along $k_2$: therefore $k_3 \mapsto \theta_0(k_2,k_3)$ can be choosen to be periodic, compare \eqref{eqn:periodicargument}. Iterating this process as in in the proof of Theorem \ref{thm:2D}.\ref{item:2D_alphabeta}, we conclude that we can write
\[ f(k_2,k_3) = \eu^{2 \pi \iu \, \theta(k_2,k_3)} \, f(0,k_3) \]
where $\theta(k_2,k_3)$ is a continuous function of both arguments and is moreover periodic in $k_3$. By computing the winding number of $f$ along the direction $k_2$ for fixed $k_3$, which vanishes by hypothesis, we can show that $\theta$ is also periodic in $k_2$, compare again \eqref{eqn:periodicargument}. Since now also $f(0,k_3)$ has a vanishing winding number along $k_3$, we conclude by Theorem \ref{thm:2D}.\ref{item:2D_alphabeta} that it is itself of the form
\[ f(0,k_3) = \eu^{2 \pi \iu \, \widetilde{\theta}(k_3)} \, f(0,0) \]
with $\widetilde{\theta}$ continuous and periodic. We can finally conclude that
\[ f_t(k_2, k_3) := \eu^{2 \pi \iu \, (1-t) \, \theta(k_2,k_3)} \,\eu^{2\pi\iu \, (1-t) \, \widetilde{\theta}(k_3)} \, f(0,0)^{1-t}, \quad t \in [0,1], \]
continuously deforms $f$ through periodic maps to the constant map equal to $1$, as claimed.
\end{proof}

\begin{remark}[Normal form for the 3D matching matrices] \label{rmk:normalalpha3D}
Combining the last two statements in Theorem \ref{thm:3D}, we conclude that the family of unitary matrices $\alpha$ in \eqref{eqn:3Dmatching} can be continuously deformed through $\alpha\sub{2D}$-periodic families to the matrices
\begin{equation} \label{eqn:alpha3D}
\alpha\sub{3D}(k_2,k_3) := \begin{pmatrix}
\eu^{\iu\, n_2 \, k_2} \, \eu^{\iu\, n_3 \, k_3} & 0 \\
0 & \Id_{m-1}
\end{pmatrix}, \quad \text{where} \quad n_j := c_1^{\set{1,j}}(P)\in \Z, \; j \in \set{2,3}.
\end{equation}
\end{remark}

\begin{proof}[Proof of Theorem \ref{thm:main}, $d=3$]
With Theorem \ref{thm:3D} and the previous Remark at hand, the proof of our main Theorem \ref{thm:main} proceeds in the 3D case just as in the 2D case. Indeed, once again the topological obstruction, given by the three first Chern numbers, restricts the periodicity in $\bk$ of a single Bloch vector: this follows from the fact that, up to continuous deformations, the matching matrices have only a single (possibly) ``winding'' entry. If the bundle is Chern-trivial, all the Bloch vectors thus constructed are smooth and $(2\pi\Z^3)$-periodic. Otherwise, the same space-doubling trick presented in the proof of Theorem \ref{thm:main} for $d=2$ allows to construct a smooth and periodic Parseval frame of $m+1$ Bloch functions for the 3D family of projections $P(\bk)$, $\bk \in \T^3$.
\end{proof} 

\begin{remark}[Weak invariants]
As was noted in the Introduction, the 3D topological obstructions are encoded in three first Chern numbers $c_1^{\set{1,2}}(P)$, $c_1^{\set{2,3}}(P)$ and $c_1^{\set{1,3}}(P)$, which are inherently 2D objects; as such, they are dubbed ``weak invariants'' in the physics literature, to stress their lower dimensionality with respect to the dimension of the system.
\end{remark}

\section{4D case and the second Chern number}

We finally arrive at the 4D case, so in this Section $P(\bk)$, $\bk \in \T^4$, will denote a $(2\pi\Z^4)$-periodic family of rank-$m$ projections. As always our aim is to construct spanning Bloch functions for these projections, and once again we start from the 3D restriction $P(0,k_2,k_3,k_4)$, $(k_2,k_3,k_4) \in \T^3$. In view of the results of the previous Section, this family can be spanned by smooth orthonormal Bloch functions $\set{\phi_a(0,k_2,k_3,k_4)}_{1 \le a \le m}$ which have $\alpha\sub{3D}(k_3,k_4)$ in \eqref{eqn:alpha3D} as their matching matrices. Explicitly, this means that all these Bloch functions are $(2\pi\Z^3)$-periodic but for the first one, which instead satisfies
\begin{equation} \label{eqn:4D_phi}
\begin{gathered}
\phi_1(0, k_2, k_3, k_4 + 2\pi) = \phi_1(0, k_2, k_3, k_4), \\
\phi_1(0, k_2, k_3 + 2\pi, k_4) = \eu^{\iu \, c_1^{\set{3,4}}(P) \, k_4} \, \phi_1(0, k_2, k_3, k_4), \\
\phi_1(0, k_2 + 2\pi, k_3, k_4) = \eu^{\iu \, c_1^{\set{2,3}}(P) \, k_3} \, \eu^{\iu \, c_1^{\set{2,4}}(P) \, k_4} \, \phi_1(0, k_2, k_3, k_4).
\end{gathered}
\end{equation}

Once more we extend smoothly these Bloch functions in the $k_1$-direction by means of parallel transport, and define
\begin{equation} \label{eqn:4D_psi}
\psi_a(k_1,k_2,k_3,k_4) := T_{(k_2,k_3,k_4)}(k_1) \, \phi_a(0,k_2,k_3,k_4), \quad 1 \le a \le m.
\end{equation}
Upon direct inspection, one checks that

\begin{lemma} \label{lemma:4Dalpha}
The matching matrices for the Bloch vectors \eqref{eqn:4D_psi}, defined as
\begin{equation} \label{eqn:4Dmatch}
\alpha_{ab}(k_2,k_3,k_4) := \left\langle \phi_b(0,k_2,k_3,k_4), \, T_{(k_2,k_3,k_4)}(2\pi) \, \phi_a(0,k_2,k_3,k_4) \right\rangle, \quad 1 \le a, b \le m,
\end{equation}
satisfy the following properties:
\begin{enumerate}
 \item $\alpha(k_2,k_3,k_4)$ is unitary;
 \item $\alpha(k_2,k_3,k_4)$ is $(2\pi\Z)$-periodic in $k_4$, that is,
 \[ \alpha(k_2,k_3,k_4+2\pi) = \alpha(k_2,k_3,k_4); \]
 \item $\alpha(k_2,k_3,k_4)$ is $\alpha\sub{2D}$-periodic in $k_3$, that is,
 \[ \alpha(k_2,k_3+2\pi,k_4) = \alpha\sub{2D}(k_4)^{-1} \, \alpha(k_2,k_3,k_4) \, \alpha\sub{2D}(k_4), \]
 with $\alpha\sub{2D}$ defined as in \eqref{eqn:alpha2D};
 \item $\alpha(k_2,k_3,k_4)$ is $\alpha\sub{3D}$-periodic in $k_2$, that is,
 \[ \alpha(k_2+2\pi,k_3,k_4) = \alpha\sub{3D}(k_3,k_4)^{-1} \, \alpha(k_2,k_3,k_4) \, \alpha\sub{3D}(k_3,k_4), \]
 with $\alpha\sub{3D}$ defined as in \eqref{eqn:alpha3D}.
\end{enumerate}
\end{lemma}

Notice that $\alpha\sub{2D}$-periodicity in $k_3$ and $\alpha\sub{3D}$-periodicity in $k_2$ are compatible with each other, as $\alpha\sub{3D}$ itself is $\alpha\sub{2D}$-periodic in $k_4$ -- actually, for the explicit diagonal matrices in \eqref{eqn:alpha2D} and \eqref{eqn:alpha3D}, $\alpha\sub{2D}$-periodicity reduces to mere periodicity. In order to have a more concise terminology, we will say that a family of unitary matrices which satisfies the conditions specified in the previous Lemma is \emph{pseudo-periodic} in its coordinates $(k_2, k_3, k_4)$.

As in the previous Sections, we can reduce the study of periodic smooth Bloch functions to the homotopy properties of pseudo-periodic families of unitary matrices, by means of the following

\begin{theorem} \label{thm:beta4D}
\begin{enumerate}
 \item If the orthonormal basis $\set{\psi_a(k_1,k_2,k_3,k_4)}_{1 \le a \le m}$ has matching matrices $\alpha(k_2,k_3,k_4)$, and if $\beta(k_1,k_2,k_3,k_4)$ is a family of $m \times m$ unitary matrices with $\beta(0,k_2,k_3,k_4) \equiv \Id$ which is pseudo-periodic in $(k_2,k_3,k_4)$, then the orthonormal basis
\[ \phi_a(\bk) := \sum_{b=1}^{m} \psi_b(\bk) \, \beta_{ba}(\bk), \quad \bk = (k_1,k_2,k_3,k_4), \]
has matching matrices
\begin{equation} \label{eqn:beta4D}
\widetilde{\alpha}(k_2,k_3,k_4) := \beta(k_1,k_2,k_3,k_4) \, \alpha(k_2,k_3,k_4) \, \beta(k_1+2\pi,k_2,k_3,k_4)^{-1}.
\end{equation}
 \item Given two $U(m)$-valued pseudo-periodic maps $\alpha, \widetilde{\alpha}$, the following are equivalent.
 \begin{enumerate}
 \item There exist $\beta(\bk) \in U(m)$ as in the previous point such that \eqref{eqn:beta4D} holds.
 \item The maps $\alpha, \widetilde{\alpha}$ are \emph{pseudo-periodically homotopically equivalent}, that is, they can be continuously deformed one into the other through maps which are pseudo-periodic in $(k_2,k_3,k_4)$.
 \end{enumerate}
\end{enumerate}

In particular, the family of projections $P(\bk)$, $\bk \in \T^4$, admits a orthonormal basis of smooth and periodic Bloch functions if and only if the family of matching matrices in \eqref{eqn:4Dmatch} is pseudo-periodically homotopic to the constant map $\widetilde{\alpha} \equiv \Id_m$.
\end{theorem}
\begin{proof}
Once again the proofs presented in Ref.\ \onlinecite[Propositions~5.1 and~6.3]{cornean2019parseval} apply almost \emph{verbatim}. The first statement in the Theorem follows from a direct computation. The second statement can be argued exactly as in the proof of Theorem \ref{thm:2D}.\ref{item:2D_alphabeta}, minding that all matrices involved are now pseudo-periodic. The final statement immediately follows from the previous two. 
\end{proof}

\subsection{(Pseudo-)periodic homotopy theory}

The previous result drives us to study the pseudo-periodic homotopy classes of families of matching matrices $\alpha(k_2, k_3, k_4)$. While in lower-dimensional cases the matching matrices could be deformed to their $U(1)$-part, we will show here that their $SU(m)$-part plays a crucial role as well. Following the convention of \eqref{eqn:U(1)xSU(m)}, let us write for $\bk = (k_2, k_3, k_4)$
\begin{equation} \label{eqn:alpha=deltasigma}
\alpha(\bk) = \delta(\bk) \, \sigma(\bk), \quad \text{where} \quad \delta(\bk) \equiv \delta(\alpha(\bk)) = \begin{pmatrix} \det \alpha(\bk) & 0 \\ 0 & \Id_{m-1} \end{pmatrix} \quad \text{and} \quad \sigma(\bk) \equiv \sigma(\alpha(\bk)) \in SU(m).
\end{equation}
Notice that, for a pseudo-periodic $\alpha$, the above $\delta$ is $(2\pi\Z^3)$-periodic while $\sigma$ is itself pseudo-periodic, since $\alpha\sub{2D}$ in \eqref{eqn:alpha2D} and $\alpha\sub{3D}$ in \eqref{eqn:alpha3D} are diagonal and hence commute with the diagonal matrix $\delta$. Since the above decomposition holds for all pseudo-periodic $\alpha$'s, it is also clear that $\alpha$ and $\widetilde{\alpha}$ are pseudo-periodically homotopic to each other if only if the corresponding $\delta$ and $\widetilde{\delta}$ are homotopic (as periodic maps) and in addition the corresponding $\sigma$ and $\widetilde{\sigma}$ are pseudo-periodically homotopic (as $SU(m)$-valued maps). 

The homotopy theory of maps $\delta(\bk)$ of the type above is deduced from the following

\begin{proposition}
Two maps $f,\widetilde{f} \colon \T^3 \to U(1)$ are homotopic if and only if
\[ \odeg_j(f) = \odeg_j (\widetilde{f}) \quad \text{for all } j \in \set{2,3,4}. \]
\end{proposition}
\begin{proof}
The argument follows closely the one presented at the end of the proof of Proposition \ref{prop:mu-nu-periodic}, where we showed that $f \colon \T^2 \to U(1)$ is characterized up to homotopy by the winding numbers of its restrictions to $\T^1_j \subset \T^2$, $j \in \set{2,3}$; in this case, there is just with one extra dimension, leading to one extra $1$-degree, to take into account. We leave the details to the reader.
\end{proof}

\begin{remark}[Normal form for the 4D matching matrices, $U(1)$-part] \label{rmk:normaldelta4D}
In view of Remark \ref{rmk:winding} and Theorem \ref{thm:3D}.\ref{item:3Ddeg=c1}, we have
\[ \odeg_j(\alpha) = c_1^{\set{1,j}}(P)\,. \]
In particular, the map $\delta \colon \T^3 \to U(m)$ appearing in \eqref{eqn:alpha=deltasigma} is homotopic to
\begin{equation} \label{eqn:delta4D}
\delta\sub{4D}(k_2,k_3,k_4) = \begin{pmatrix}
\eu^{\iu \, n_2 \, k_2} \, \eu^{\iu \, n_3 \, k_3} \, \eu^{\iu \, n_4 \, k_4} & 0 \\ 0 & \Id_{m-1}
\end{pmatrix}, \quad \text{where} \quad n_j := c_1^{\set{1,j}}(P) \in \Z.
\end{equation}
\end{remark}

\begin{corollary} \label{cor:boundary}
Any pseudo-periodic $U(m)$-valued map $\alpha(\bk)$, $\bk = (k_2,k_3,k_4) \in \R^3$, is pseudo-periodically homotopic to a map $\widetilde{\alpha}(\bk)$ such that
\[ \widetilde{\alpha} \Big|_{\partial \Omega} \equiv \Id_m, \quad \text{where} \quad \Omega := [0,2\pi]^3\,. \]
\end{corollary}
\begin{proof}
In view of the above Remark, up to continuous deformation it can be assumed that $\delta(\alpha(\bk))$ is in the form \eqref{eqn:delta4D}: in particular, it is equal to $\Id_m$ on the boundary of the ``pseudo-periodicity cell'' $\Omega = [0,2\pi]^3$. Thefore, it remains to show that the same holds for the map $\sigma(\bk) \equiv \sigma(\alpha(\bk)) \in SU(m)$.

For $j \in {2,3,4}$, let $\Omega_j \subset \partial \Omega$ be the face of the cube obtained by freezing $k_j = 0$. The restriction $\sigma^{(j)} := \sigma \Big|_{\Omega_j}$ is a $(\mu,\nu)$-periodic map, in the sense of Proposition \ref{prop:mu-nu-periodic}, for appropriate periodic maps $\mu, \nu \colon \T^1 \to U(m)$ which coincide with either $\alpha\sub{2D}$ or with appropriate restrictions of $\alpha\sub{3D}$. Since $\Omega_j \simeq [0,2\pi]^2$, the above-mentioned Proposition implies that these restrictions can be brought to diagonal form as in \eqref{eqn:alpha-tildealpha-mu-nu} through \emph{$U(m)$-valued} maps. Notice however that $\det \sigma \equiv 1$, and therefore the determinant of $\sigma$ (and of its restrictions) does not wind along any direction: this implies that the determinant continues not to wind along the whole deformation, and therefore the deformation itself can be achieved through \emph{$SU(m)$-valued} maps by possibly combining it with a deformation ``unwinding'' the determinant. This means that each $\sigma^{(j)}$ can be deformed continuously, trough $SU(m)$-valued maps $\sigma^{(j)}_s$ depending continously on $s \in [0,1]$, to the map on $\Omega_j$ which is constant equal to $\Id_m$ (since it must be both diagonal and in $SU(m)$). Making sure to choose continuous junctions across the edges of $\partial \Omega$, and by imposing the appropriate pseudo-periodicity, this deformation can be extended to the whole boundary of $\Omega$.

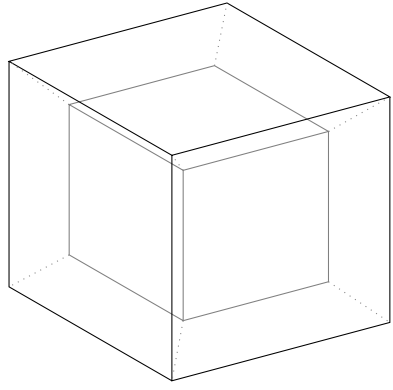
\begin{figure}[htb]
\centering
\begin{tikzpicture}[>=stealth] 
\node (o) at (0,1) {};
\path (o)+(15:2cm) node (a) {}
      (o)+(150:1.75cm) node (b) {}
      (o)+(0,-2) node (o1) {}
      (b)+(15:2cm) node (c) {}
      (o1)+(15:2cm) node (a1) {}
      (o1)+(150:1.75cm) node (b1) {};
\node (o2) at (-0.15,1.2) {};
\path (o2)+(15:3cm) node (a2) {}
      (o2)+(150:2.5cm) node (b2) {}
      (o2)+(0,-3) node (o3) {}
      (b2)+(15:3cm) node (d) {}
      (o3)+(15:3cm) node (a3) {}
      (o3)+(150:2.5cm) node (b3) {};
\draw [gray] (o.center) -- (a.center) -- (c.center) -- (b.center) -- (o.center) -- (o1.center) -- (a1.center) -- (a.center)
      (o1.center) -- (b1.center) -- (b.center);
\draw (o2.center) -- (a2.center) -- (d.center) -- (b2.center) -- (o2.center) -- (o3.center) -- (a3.center) -- (a2.center)
      (o3.center) -- (b3.center) -- (b2.center);
\draw [gray, dotted] (o.center) -- (o2.center)
      (o1.center) -- (o3.center)
      (a.center) -- (a2.center)
      (b.center) -- (b2.center)
      (a1.center) -- (a3.center)
      (b1.center) -- (b3.center)
      (c.center) -- (d.center);
\end{tikzpicture}
\caption{``Thickened'' cell $[-1,2\pi+1]^3$, containing $\Omega = [0,2\pi]^3$ (drawn in gray)}
\label{fig:thicc3D}
\end{figure}

We can now use a ``thickening'' trick similar to the one in the proof of Proposition \ref{prop:mu-nu-periodic}, but now for the 3D cell $\Omega$ (see Figure~\ref{fig:thicc3D}). Extend the definition of $\sigma$ to a larger cell $[-1,2\pi+1]^3$ by using $s \in [0,1]$ as a ``radial coordinate'' bulging out of the faces $\Omega_j$; then, let $\sigma_t$ be the restriction of this extension to $[-t,2\pi+t]^3$, once this enlarged cell is appropriately rescaled to $\Omega$. The end-point of this continuous deformation of $\sigma$ yields a map $\widetilde{\sigma}$ on $\Omega$ which is $SU(m)$-valued and attains the value $\Id_m$ on the whole $\partial \Omega$. The pseudo-periodic extension of this map from the cell $\Omega$ to the whole $\R^3$, once it is multiplied by $\delta\sub{4D}$ in \eqref{eqn:delta4D}, yields the desired deformation of $\alpha$.
\end{proof}

In view of Theorem \ref{thm:beta4D} and of the above Corollary \ref{cor:boundary}, the family of projections $P(\bk)$, $\bk \in \T^4$, can be spanned by smooth orthonomal Bloch functions
\begin{equation} \label{eqn:normalizedBFs}
\set{\phi_1(k_1,k_2,k_3,k_4), \ldots, \phi_m(k_1,k_2,k_3,k_4)}
\end{equation}
whose matching matrices $\alpha(\bk) = \delta(\bk) \, \sigma(\bk)$ have $\delta(\bk) = \delta\sub{4D}(\bk)$ as in \eqref{eqn:delta4D} while $\sigma(\bk) \in SU(m)$ depends pseudo-periodically on $\bk$ and is such that 
\begin{equation} \label{eqn:boundary_sigma}
\sigma \Big|_{\partial \Omega} \equiv \Id_m.
\end{equation}
We will say concisely that such $\sigma$'s define \emph{normalized} pseudo-periodic $SU(m)$-valued maps.

\begin{remark} \label{rmk:S3}
Topologically, a $SU(m)$-valued map on $\Omega$ which satisfies the normalization \eqref{eqn:boundary_sigma} is tantamount to a map defined on the $3$-sphere $S^3$. Indeed, the boundary of the cell can be brought to infinity and then compactified to a point, where the map attains the value $\Id_m$; the resulting map is then defined on the one-point compactification of $\R^3$, which is indeed $S^3$.
\end{remark}

The next step in the study of homotopy classes of such maps consists in finding a ``normal form'' of block-diagonal type. Let us first make the obvious remark that, if $m=1$, then necessarily $\sigma(\bk) \equiv 1$ and there is nothing to discuss. We will therefore always assume $m \ge 2$ hereinafter.

\begin{theorem} \label{thm:SU(m)->SU(2)}
Let $m \ge 2$ and let $\sigma(\bk)$, $\bk = (k_2,k_3,k_4)$, be a normalized pseudo-periodic $SU(m)$-valued map. Then one can construct an homotopy, which preserves pseudo-periodicity and the normalization \eqref{eqn:boundary_sigma}, between
\begin{equation} \label{eqn:sigma_eta} 
\sigma(\bk) \quad \text{and} \quad \widetilde{\sigma}(\bk) := \begin{pmatrix} \eta(\bk) & 0 \\ 0 & \Id_{m-2} \end{pmatrix}, \quad \eta(\bk) \in SU(2),
\end{equation}
where $\eta(\bk)$ defines a normalized pseudo-periodic $SU(2)$-valued map.
\end{theorem}
\begin{proof}
Let us first restrict $\sigma$ to $\Omega$. In view of \eqref{eqn:boundary_sigma}, this restriction also admits a  \emph{$(2\pi\Z^3)$-periodic} extension, for which the column interpolation argument of Ref.\ \onlinecite{gontier2019numerical} applies. As was briefly explained in the proof of Theorem \ref{thm:2D}.\ref{item:2D_alphabeta}, the argument deforms continously the columns of the matrix to constant vectors. The construction is applied inductively, starting from $\widetilde{m} = m$ and going down in $\widetilde{m}$ for as long as the dimension of the sphere in which these columns lie, namely $S^{2\widetilde{m}+1} \subset \C^{\widetilde{m}}$, is bigger than the dimensionality of the vector of parameters $\bk$ on which the matrix depends, namely $3$ in our case. We have that the column interpolation argument fails only when $2 \widetilde{m} + 1 \le 3$, or $\widetilde{m} \le 2$. This yields that $m-2$ columns of the matrix $\sigma$ can be continously deformed to be the last standard basis vectors in $\C^m$, and therefore we end up with a block-diagonal form for $\sigma$ of the type claimed in the statement -- once again, the argument a priori produces a deformation within $U(m)$, but the determinant doesn't wind, and therefore it can be deformed to $1$ if needed. Since $\sigma$ is already equal to $\Id_m$ on the boundary of~$\Omega$, we can assume that such deformation does not modify the values on $\partial \Omega$. Once pseudo-periodicity is imposed to extend the definition from $\Omega$ to $\R^3$, the proof is concluded.
\end{proof}

The above result reduces the homotopy theory of normalized pseudo-periodic $SU(m)$-valued maps $\sigma(\bk)$, with $m \ge 2$, to that of $SU(2)$-valued maps $\eta(\bk)$ with the same properties. 

\begin{theorem} \label{thm:homotopy_3deg}
Two normalized pseudo-periodic $SU(2)$-valued maps $\eta, \widetilde{\eta}$ are pseudo-periodically homotopic if and only if
\begin{equation} \label{eqn:3deg}
\tdeg(\eta) = \tdeg\big(\widetilde{\eta}\big), \quad \text{where} \quad \tdeg(g) := \frac{1}{24 \, \pi^2} \int_{\T^3} \Tr_{\C^2} \left[ \left( g^{-1} \, \di g \right)^{\wedge 3} \right]\, \in \Z \,.
\end{equation}
\end{theorem}

\begin{remark}[$3$-degree and topological degree of $SU(2)$-valued maps] \label{rmk:3deg_topo}
The quantity $\tdeg(g)$ appearing in \eqref{eqn:3deg} is called the \emph{$3$-degree}, and characterizes the \emph{topological degree} of smooth maps between manifolds of dimension $3$, specifically in the case where the target manifold is $SU(2)$, which is topologically a $3$-sphere \cite{mukherjee2015differential, manton2004topological}. Many properties of this topological invariant are collected for the readers' convenience in Appendix \ref{app:3deg}. In particular, in view of Corollary \ref{cor:3deg(sigma)=3deg(eta)}, the $3$-degree of $\eta$ appearing in \eqref{eqn:3deg} can be computed directly from $\alpha$ by means of a similar integral, with the only difference that the trace is computed over the whole $\C^m$ rather than on $\C^2$: this integral will be also dubbed the $3$-degree of $\alpha$.

The $3$-degree of certain unitary-valued maps appears also in the field-theoretic investigation of the topological properties of 2D Floquet insulators, see Ref.s \onlinecite{carpentier2015construction, monaco2017gauge} and references therein.
\end{remark}

\begin{proof}[Proof of Theorem \ref{thm:homotopy_3deg}]
Proposition \ref{prop:3deg} shows that, even when $\eta$ is only pseudo-periodic, the differential form used to compute the $3$-degree has an integral which does not depend on the cell chosen for the dual lattice $\Gamma^*$ (hence the notation of integration over $\T^3$), and that moreover homotopic (normalized, pseudo-periodic) maps have the same $3$-degree. By the additivity of such degree, which is also shown in Proposition \ref{prop:3deg}, in order to show conversely that two maps with the same $3$-degree can be deformed one into the other it suffices to show that a normalized pseudo-periodic map $f(\bk) \in SU(2)$ with $\tdeg(f) = 0$ can be deformed to the constant map $\Id_2$. In view of Remark~\ref{rmk:S3}, we will rather consider $f$ as defined over the $3$-sphere $S^3$, parametrized by coordinates $\bx = (x_0, x_1, x_2, x_4)$ with $x_0^2 + \cdots + x_4^2=1$, with values in $SU(2)$, which is also topologically a $3$-sphere. The general statement that null-degree maps $S^n \to S^n$ are homotopic to a constant map is the content of the so-called \emph{Hopf's degree theorem}: see e.g.\ Ref.\ \onlinecite[Theorem 6.6.6]{mukherjee2015differential} for a proof, which is inductive in the dimension $n$ of the spheres.

We sketch here the construction of such an homotopy; in order not to overburden the argument, we provide here only some details, as fully explicitating the construction of the required homotopy in a concise way proves challenging. 

Pick a regular value $y \in SU(2)$ for $f \colon S^3 \to SU(2)$ (they form a full-measure set in $SU(2)$ by Sard's lemma); we can choose $y = - \Id_2$ without loss of generality. We consider then two cases.

\begin{description}
\item[Case 1] the point $-\Id_2 \in SU(2)$ lies outside of the image of $f$: this is the generic case for such null-degree maps. In this situation, we can perform a stereographic projection $\pi$ of the range of $f$ from the antipodal point $\Id_2 \in SU(2) \simeq S^3$, and contract the projection $\pi(f(\bx))$ to the origin in $\R^3$ by a rescaling $(1-t) \, \pi(f(\bx))$. The inverse stereographic projection yields the desired deformation with values in $S^3$.%\
\footnote{Alternatively, we can use the fact that the exponential map $\exp \colon B_{r}(0) \subset \mathfrak{su}(2) \to SU(2) \setminus \{-\Id_2\}$ is a diffeomorphism, where $r$ is the injectivity radius of the $3$-sphere (see e.g.\ Ref.\ \onlinecite[Chap.~5]{petersen1988riemannian}). This means that one can write $f(\bk) := \eu^{\iu \, \vec{n}(\bk) \cdot \vec{\tau}}$, where $\vec{\tau} = (\tau_1, \tau_2, \tau_3)$ are the three Pauli matrices and $\vec{n}(\bk) \in \R^3$ is a smooth map which is compatible with the pseudo-periodicity of $f$ (it can be defined via a matrix logarithm which is smooth on $SU(2) \setminus \{-\Id_2\}$). Therefore, $f_t(\bk) := \eu^{\iu \, (1-t)\, \vec{n}(\bk) \cdot \vec{\tau}}$ deforms $f(\bk)$ to $\Id_2$ continuously via pseudo-periodic maps.}

\item[Case 2] the preimage $f^{-1} \set{-\Id_2}$ is non-empty (non-generic case). In this case, $f$ can be deformed into a map falling under the previous Case 1. One can first show that $f$ can be continuously deformed to a function $g$ such that the preimage of interest lies in a ``great circle'' $S^1 \subset S^3$, and that said map $g$ can be more easily deformed to a constant, potentially by considering the restriction $g \big|_{S^1} \colon S^1 \to g(S^1) \subset S^3$, and then ``straightening out'' the loop $g(S^1) \subset S^3$ to be a circle through $-\Id_2$, so as to consider the 1-degree of the resulting $S^1 \to S^1$ map. A homotopy unwinding $g\big|_{S^1}$ might then be lifted to $g$ as a whole.
\end{description}

For the interested reader, we detail below how to constrain the preimages of a given regular value of $f$ to a circle $S^1$. Let us first enumerate the points in this preimage as $\set{\bx_1, \ldots, \bx_p}$. Let $M \subset S^3$ be the subset of the domain of $f$ constructed as follows: Pair the points in $f^{-1}\set{-\Id_2}$ in all possible ways; for any such pair $\set{\bx_i,\bx_j}$, consider the plane in $\R^4$ which contains $\bx_i$, $\bx_j$ and the origin (if the points $\bx_i$ and $\bx_j$ happen to be antipodal on $S^3$ and hence collinear with the origin, pick any plane which contains all three points); $M$ is then the union of the intersections of all such planes with $S^3$. As a union of $1$-dimensional circles, $M$ is of measure zero in $S^3$; therefore we can pick a generic point $\bx_0 \in S^3 \setminus M$, such that also $-\bx_0 \in S^3$ lies outside of $M$. Consider then the equatorial $S^2 \subset S^3$ subject to the choice of ``poles'' $\set{\bx_0, -\bx_0}$. By construction, with this choice of ``poles'', no two points in $f^{-1}\{-\Id_2\}$ lie on the same ``meridian'', that is, they cannot lie on the same great circle which passes through $\bx_0$ as well. Let us then project $\set{\bx_1, \ldots, \bx_p}$ onto the equator $S^2$ along the ``meridians'': these projections $\set{\bx_1', \ldots, \bx_p'}$ will again be distinct. We now consider a tubular neighbourhood in $S^3$ of the equator $S^2$, large enough to contain all preimages $\set{\bx_1, \ldots, \bx_p}$. A pictorial image of this configuration in cylindrical coordinates is provided in Figure~\ref{fig:mercator}. We define then a family of smooth maps $\phi_t \colon S^3 \to S^3$ such that $\phi_{t=0}$ is the identity, and which displaces the points $\set{\bx_1, \ldots, \bx_p}$ to the points $\set{\bx_1', \ldots, \bx_p'}$ following meridians. More specifically it displaces a given loop passing through $\set{\bx_1, \ldots, \bx_p}$ within the tubular neighbourhood to the equator $S^2$, and leaves the complement of this tubular neighbourhood invariant. The family thus culminates at $t=1$ with a map in which the aforementioned loop now lies on the equator $S^2 \subset S^3$. Each map $\phi_t : S^3 \to S^3$ is bijective, enabling us to define the composition $f_t := f \circ \phi_t^{-1}$ to obtain a smooth deformation of the original map $f = f_{t=0}$ with a map $f_{t=1}$ which still has $-\Id_2$ as a regular value, but such that the preimage $f_{t=1}^{-1}\set{-\Id_2}$ is constrained on the equator $S^2 \subset S^3$.

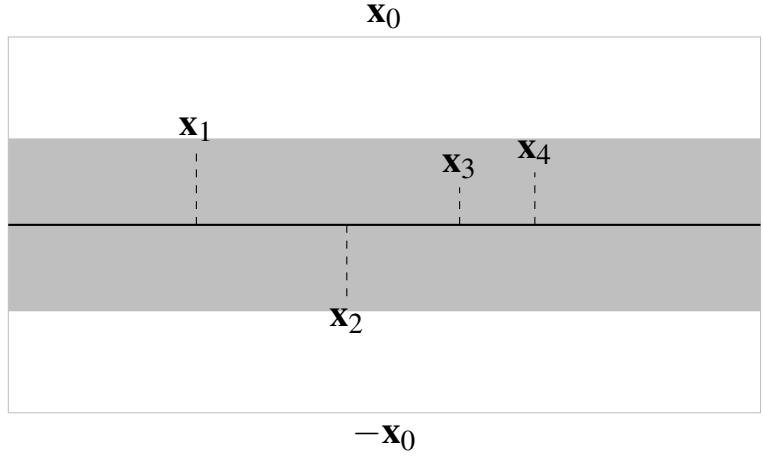
\begin{figure}[htb]
\centering
\begin{tikzpicture}[>=stealth]
\draw [lightgray] (-5,-2.5) rectangle (5,2.5);
\draw [line width=23mm, lightgray] (-5, 0) -- (5,0);
\draw [thick] (-5, 0) -- (5,0);
\draw [dashed] 
      (-2.5,0) -- (-2.5,1) node [anchor=south] {\Large $\bx_1$}
      (-.5,0) -- (-.5,-1) node [anchor=north] {\Large $\bx_2$}
      (1,0) -- (1,0.5) node [anchor=south] {\Large $\bx_3$}
      (2,0) -- (2,0.7) node [anchor=south] {\Large $\bx_4$};
\draw (0,2.5) node [anchor=south] {\Large $\bx_0$}
      (0,-2.5) node [anchor=north] {\Large $-\bx_0$};
\end{tikzpicture}
\caption{Representation of $S^3$, resp.\ $S^2$, in cylindrical coordinates; equatorial $S^2$, resp.\ $S^1$ (thick straight line); points $\set{\bx_1, \ldots, \bx_p}$ in the preimage of a regular value have different projections (dashed lines) on the equator; a tubular neighbourhood of the equator and of such projections (shaded region)}
\label{fig:mercator}
\end{figure}

We now repeat the same type of argument, but to the restriction $f_{t=1} \big|_{S^2} \colon S^2 \to S^3$. Again, the generic choice of an equator $S^1 \subset S^2$ yields distinct projections to $S^1$ along ``meridians'' of the points $\set{\bx_1', \ldots, \bx_p'}$; by retracting along meridians, these preimages can be then pushed further to $S^1$. The reader can once again refer to Figure \ref{fig:mercator}. By smearing out this deformation along the transverse coordinates to $S^2 \subset S^3$, we can finally arrive at the desired map $g \colon S^3 \to S^3$, obtained from $f$ by a continuous deformation (which in particular doesn't change the value of the $3$-degree), and such that $g^{-1}\set{-\Id_2}$ is contained in a circle $S^1 \subset S^3$.  
\end{proof}

\subsection{Topological degree and 2nd Chern class}

The final link between the pseudo-periodic homotopy theory of the family of matching matrices and the topological invariants associated to the projections $P(\bk)$, $\bk \in \T^4$, is provided by the next

\begin{theorem} \label{thm:3deg=c2}
If $\alpha(k_2,k_3,k_4)$ denote the matching matrices of the Bloch functions \eqref{eqn:normalizedBFs}, it holds that
\[ \tdeg(\alpha) = - c_2^{\set{1,2,3,4}}(P) \in \Z\,. \]
\end{theorem}
\begin{proof}
The proof is similar in spirit to the one of Theorem \ref{thm:2D}.\ref{item:2D_alphadelta}, and is based on the following observation which is shown in Proposition \ref{prop:CS}: Like the first Chern form has a local primitive given by the trace of the Berry connection $1$-form \eqref{eqn:BerryConnection}, the second Chern form $c_2(P)$ in \eqref{eqn:c12} can be locally expressed as
\[ c_2(P) = \di CS, \quad \text{where} \quad CS := \frac{1}{2} \left[ \Tr(A) \wedge \Tr(F) - \Tr \left( F \wedge A - \frac{2\pi\iu}{3} \, A \wedge A \wedge A \right) \right]\,. \]
The $3$-form $CS$ is called \emph{Chern--Simons form}, and since it is expressed in terms of the Berry connection $A$ it requires an orthonormal basis of Bloch functions to be defined. The dependence of $CS$ on the chosen gauge is as follows: if $\Psi = \set{\psi_a(\bk)}_{1 \le a \le m}$ is obtained from $\Phi = \set{\phi_a(\bk)}_{1 \le a \le m}$ by means of a $U(m)$-valued change of basis $\gamma(\bk)$, then
\[ CS^{(\Psi)} = CS^{(\Phi)} - \frac{1}{24\pi^2} \, \Tr_{\C^m} \left[ \left( \gamma^{-1} \, \di \gamma \right)^{\wedge 3} \right] - \di B_{\gamma}, \quad \text{where} \quad B_{\gamma} := \frac{1}{4\pi\iu} \left[ \Tr_{\C^m}\left(A \wedge \di \gamma \, \gamma^{-1}\right) + \Tr_{\C^m}(A) \wedge \Tr_{\C^m}\left(\di \gamma \, \gamma^{-1} \right) \right]\,. \]

Let us then use this information, and compute the second Chern number $c_2^{\set{1,2,3,4}}(P)$ by means of the Bloch basis \eqref{eqn:normalizedBFs}, which is defined on the whole space but is not periodic in general. Nevertheless we have by Stokes' theorem
\begin{align*}
c_2^{\set{1,2,3,4}}(P) & = \int_{\T^4} c_2(P) = \int_{[0,2\pi]^4} \di CS = \int_{\partial [0,2\pi]^4} CS \\
& = \int_{\set{k_1=2\pi}} CS - \int_{\set{k_1=0}} CS + \int_{\set{k_2=2\pi}} CS - \int_{\set{k_2=0}} CS + \int_{\set{k_3=2\pi}} CS - \int_{\set{k_3=0}} CS \, .
\end{align*}
Since the Bloch functions \eqref{eqn:normalizedBFs} are periodic in $k_4$, the extra boundary term that we would have in the above equality, namely the difference of the integrals of $CS$ on $\set{k_4=2\pi}$ and $\set{k_4=0}$, vanishes. This periodicity in $k_4$ will be used repeatedly also in the following. 

On the other hand, the Bloch functions fail to be periodic in the other directions: the lack of periodicity in $k_1$ is measured by the matching matrices $\alpha(k_2,k_3,k_4)$, the lack of periodicity in $k_2$ is measured by $\alpha\sub{3D}(k_3,k_4)$, and the lack of periodicity in $k_3$ is measured by $\alpha\sub{2D}(k_4)$ (compare \eqref{eqn:4D_phi} and \eqref{eqn:4Dmatch}). In view of the above gauge-dependence of the Chern--Simons form, this implies that
\begin{equation} \label{eqn:boundaryCS}
\begin{aligned}
\int_{\set{k_1=2\pi}} CS - \int_{\set{k_1=0}} CS & = - \frac{1}{24\pi^2} \, \int_{\set{k_1=0}} \Tr_{\C^m} \left[ \left( \alpha^{-1} \, \di \alpha \right)^{\wedge 3} \right] - \int_{\set{k_1=0}} \di B_{\alpha}\,, \\
\int_{\set{k_2=2\pi}} CS - \int_{\set{k_2=0}} CS & = - \frac{1}{24\pi^2} \, \int_{\set{k_2=0}} \Tr_{\C^m} \left[ \left( \alpha\sub{3D}^{-1} \, \di \alpha\sub{3D} \right)^{\wedge 3} \right] - \int_{\set{k_2=0}} \di B_{\alpha\sub{3D}}\,, \\
\int_{\set{k_3=2\pi}} CS - \int_{\set{k_3=0}} CS & = - \frac{1}{24\pi^2} \, \int_{\set{k_3=0}} \Tr_{\C^m} \left[ \left( \alpha\sub{2D}^{-1} \, \di \alpha\sub{2D} \right)^{\wedge 3} \right] - \int_{\set{k_3=0}} \di B_{\alpha\sub{2D}}\,.
\end{aligned}
\end{equation}

In the last two lines, the differential forms $\Tr_{\C^m} \left[ \left( \alpha\sub{3D}^{-1} \, \di \alpha\sub{3D} \right)^{\wedge 3} \right]$ and $\Tr_{\C^m} \left[ \left( \alpha\sub{2D}^{-1} \, \di \alpha\sub{2D} \right)^{\wedge 3} \right]$ vanish identically: indeed, the matrices in the integral are diagonal and have only one non-zero entry, which contributes to the trace; but since this entry depends only on one or two of the three coordinates over which the integration takes places, the three-fold wedge product necessarily vanishes for dimensional reasons. We argue now that the integrals of $\di B_{\alpha\sub{3D}}$ and $\di B_{\alpha\sub{2D}}$ vanish as well. Indeed, let us apply Stokes' theorem once again, and get for example
\[
\int_{\set{k_3=0}} \di B_{\alpha\sub{2D}} = \int_{\partial \set{k_3=0}} B_{\alpha\sub{2D}} = \int_{\set{k_3=0, \, k_1=2\pi}} B_{\alpha\sub{2D}} - \int_{\set{k_3=0, \, k_1=0}} B_{\alpha\sub{2D}} + \int_{\set{k_3=0, \, k_2=2\pi}} B_{\alpha\sub{2D}} - \int_{\set{k_3=0, \, k_2=0}} B_{\alpha\sub{2D}}
\]
where again we neglect the matching contributions coming from $\set{k_4=2\pi}$ and $\set{k_4=0}$. Let us focus on the boundary contributions on $\set{k_1=2\pi}$ and $\set{k_1=0}$; notice that, in the definition of $B_{\alpha\sub{2D}}$, the matrix $\alpha\sub{2D} = \alpha\sub{2D}(k_4)$ is independent of $k_1$ and is therefore left unchanged, while the Berry connection $A$ changes according to \eqref{eqn:A_gauge} with a gauge transformation dictated by (the restriction to $\set{k_3=0}$ of) the matching matrices $\alpha$. Consequently we deduce that
\begin{multline*}
\int_{\set{k_3=0, \, k_1=2\pi}} B_{\alpha\sub{2D}} - \int_{\set{k_3=0, \, k_1=0}} B_{\alpha\sub{2D}} \\
= - \frac{1}{8\pi^2} \int_{\set{k_3=0=k_1}} \left[ \Tr_{\C^m} \left( \alpha^{-1} \, \di \alpha \wedge \di \alpha\sub{2D} \, \alpha\sub{2D}^{-1} \right) + \Tr_{\C^m} \left( \alpha^{-1} \, \di \alpha \right) \wedge \Tr_{\C^m} \left( \di \alpha\sub{2D} \, \alpha\sub{2D}^{-1} \right) \right] = 0
\end{multline*}
since by Corollary \ref{cor:boundary} we have $\alpha \big|_{\partial \Omega} = \alpha \big|_{\partial \set{k_1=0}} \equiv \Id_{m}$, and therefore $\di \alpha \equiv 0$ on the region of integration in the above equality. Moreover
\begin{multline*}
\int_{\set{k_3=0, \, k_2=2\pi}} B_{\alpha\sub{2D}} - \int_{\set{k_3=0, \, k_2=0}} B_{\alpha\sub{2D}} \\
= - \frac{1}{8\pi^2} \int_{\set{k_3=0=k_2}} \left[ \Tr_{\C^m} \left( \alpha\sub{3D}^{-1} \, \di \alpha\sub{3D} \wedge \di \alpha\sub{2D} \, \alpha\sub{2D}^{-1} \right) + \Tr_{\C^m} \left( \alpha\sub{3D}^{-1} \, \di \alpha\sub{3D} \right) \wedge \Tr_{\C^m} \left( \di \alpha\sub{2D} \, \alpha\sub{2D}^{-1} \right) \right] = 0
\end{multline*}
since, once again, both $\alpha\sub{2D} = \alpha\sub{2D}(k_4)$ and $\alpha\sub{3D} = \alpha\sub{3D}(k_3,k_4)$ are independent of $k_1$ and have only one possibly non-constant entry, leading to the vanishing of the wedge products appearing in the integrand above. Similar arguments apply to $B_{\alpha\sub{3D}}$, and lead to the conclusion that the last two lines in \eqref{eqn:boundaryCS} vanish. 

We are left with the equality
\[ c_2^{\set{1,2,3,4}}(P) = - \frac{1}{24\pi^2} \, \int_{\set{k_1=0}} \Tr_{\C^m} \left[ \left( \alpha^{-1} \, \di \alpha \right)^{\wedge 3} \right] - \int_{\set{k_1=0}} \di B_{\alpha} = - \frac{1}{24\pi^2} \, \int_{\set{k_1=0}} \Tr_{\C^m} \left[ \left( \alpha^{-1} \, \di \alpha \right)^{\wedge 3} \right] - \int_{\partial \set{k_1=0}} B_{\alpha} \]
again by Stokes' theorem. However, we have already argued that $\di \alpha = 0$ on the boundary of $\Omega = \set{k_1=0}$, and thefore that $B_\alpha$ vanishes identically there. The conclusion now follows from the definition of the $3$-degree of $\alpha$ in Remark \ref{rmk:3deg_topo}.
\end{proof}

\begin{remark}[Normal form for the 4D matching matrices, $SU(m)$-part]
The above results give a possible normal form for the $SU(m)$-part of the matching matrices for the Bloch functions \eqref{eqn:normalizedBFs}. For this, notice first that any $2 \times 2$ complex matrix $\eta$ can be written uniquely as
\[ \eta = \sum_{j=0}^{4} \eta_j \, \tau_j \]
where
\[ \tau_0 = \Id_2, \quad \tau_1 = \begin{pmatrix} 0 & 1 \\ 1 & 0 \end{pmatrix}, \quad \tau_2 = \begin{pmatrix} 0 & -\iu \\ \iu & 0 \end{pmatrix}, \quad \tau_1 = \begin{pmatrix} 1 & 0 \\ 0 & -1 \end{pmatrix} \]
are the Pauli matrices and
\[ \eta_j := \frac{1}{2} \, \Tr_{\C^2} \left( \tau_j \, \eta \right) \]
is the Hilbert--Schmidt scalar product of $\eta$ with $\tau_j$. The condition that $\eta \in SU(2)$ is equivalent to requiring that $c_0 := \eta_0$ and $c_j := \iu \, \eta_j$, $j \in \set{1, 2, 3}$, are real and satisfy $c_0^2 + c_1^2 + c_2^2 + c_3^2 = 1$: this realizes the isomorphism of $SU(2)$ with the $3$-sphere $S^3 \subset \R^4$. Recall also from Remark \ref{rmk:S3} that a normalized map $\sigma \colon \Omega \to SU(m)$, such that $\sigma \big|_{\partial \Omega} \equiv \Id_m$, defines a map $S^3 \simeq \R^3 \cup \set{\infty} \to SU(m)$ when $\partial \Omega$ is quotiented out to a point ``at infinity''; let us then denote by $(x_0, x_1, x_2, x_3)$, $x_i \equiv x_i(\bk)$, a set of coordinates in which the quotient space $\Omega / \partial \Omega$ can be identified by the equation $x_0^2 + x_1^2 + x_2^2 + x_3^2 = 1$.

With this notation, any normalized pseudo-periodic $SU(m)$-valued map $\sigma(\bk)$, $\bk=(k_2,k_3,k_4)$, can be brought up to homotopy to the pseudo-periodic extension of the following map on $\Omega$:
\begin{equation} \label{eqn:sigma4D}
\sigma\sub{4D}(\bk) = \begin{pmatrix} \eta\sub{4D}(\bk) & 0 \\ 0 & \Id_{m-2} \end{pmatrix} \quad \text{with} \quad \eta\sub{4D}\big(r \, \cos(\varphi), r \, \sin(\varphi), x_2, x_3\big) := r \, \cos (n \, \varphi) \, \tau_0 + r \, \sin(n \, \varphi) \, \tau_1 + x_2 \, \tau_2 + x_3 \, \tau_3 \in SU(2)
\end{equation}
where \cite{manton2004topological} $n := c_2^{\set{1,2,3,4}}(P) \in \Z$, $(r,\varphi) \in \R_+ \times [0,2\pi]$ and $r^2 + x_2^2 + x_3^2 = 1$ (compare the proof of Theorem \ref{thm:homotopy_3deg}).
\end{remark}

\subsection{Back to Bloch functions}

Collecting all previous results, we are finally able to conclude the proof of our main Theorem in the 4D situation.

\begin{proof}[Proof of Theorem \ref{thm:main}, $d=4$]
Let us assume that all the first Chern numbers and the second Chern number vanish. Then, by virtue of Remarks \ref{rmk:normalalpha2D}, \ref{rmk:normalalpha3D} and \ref{rmk:normaldelta4D}, all the matrices $\alpha\sub{2D}$, $\alpha\sub{3D}$ and $\delta\sub{4D}$ are identically equal to $\Id_m$, which means that the matching matrices $\alpha$ are periodic and lie in $SU(m)$. Moreover, Theorems \ref{thm:SU(m)->SU(2)}, \ref{thm:homotopy_3deg} and \ref{thm:3deg=c2} combine to yield that also this $SU(m)$-valued periodic family of matrices can be continuously deformed to the identity. Theorem \ref{thm:beta4D} finally shows how to construct smooth and $(2\pi\Z^4)$-periodic orthonormal Bloch functions spanning the family of projections $P(\bk)$, $\bk \in \T^4$.

Without the assumption on the vanishing of the Chern numbers, the above-mentioned results allow to conclude that the matching matrices can be brought, up to continuous deformation, to a normal form $\alpha\sub{4D} = \delta\sub{4D} \, \sigma\sub{4D}$, where $\delta\sub{4D}$ is in the form \eqref{eqn:delta4D} (being in particular diagonal with only one non-trivial entry) and $\sigma\sub{4D}$ is in the form \eqref{eqn:sigma4D} (which is also block-diagonal, with only a $2\times 2$ non-trivial block). Theorem \ref{thm:beta4D} then allows to conclude that $P(\bk)$, $\bk \in \T^4$, can be spanned by smooth orthonormal Bloch functions $\set{\widetilde{\psi}_1(\bk), \widetilde{\psi}_2(\bk), \widetilde{\psi}_3(\bk), \ldots, \widetilde{\psi}_m(\bk)}$ where the last $m-2$ are already $(2\pi\Z^4)$-periodic, while $\set{\widetilde{\psi}_1(\bk), \widetilde{\psi}_2(\bk)}$ acquire topological ``phases'' when shifting $\bk$ by vectors in $\Gamma^*$. Notice however that the projection 
\[ P_2(\bk) := \left| \widetilde{\psi}_1(\bk) \right\rangle \left\langle \widetilde{\psi}_1(\bk) \right| + \left| \widetilde{\psi}_2(\bk) \right\rangle \left\langle \widetilde{\psi}_2(\bk) \right| \]
is a \emph{periodic} subprojection of $P(\bk)$, namely $P(\bk) \, P_2(\bk) = P_2(\bk)$, and that moreover their Chern classes agree:
\[ \left[ c_1(P) \right] = \left[ c_1(P_2) \right] \,, \quad \left[ c_2(P) \right] = \left[ c_2(P_2) \right]\,. \]
Choose then another rank-$2$ projection $Q_2(\bk)$, acting possibly on some ancillary Hilbert space $\mathcal{H}'$, such that
\begin{equation} \label{eqn:Q_2}
\left[ c_1(Q_2) \right] = -\left[ c_1(P_2) \right] \,, \quad \left[ c_2(Q_2) \right] = -\left[ c_2(P_2) \right]\,. 
\end{equation}
--- see Remark \ref{rmk:Q_2} below --- and consider the family of rank-$4$ projections $P_2(\bk) \oplus Q_2(\bk)$ acting on the doubled space $\mathcal{H} \oplus \mathcal{H}'$. By virtue of what we have just shown above, this family of projections is Chern-trivial, and can be therefore spanned by four smooth and periodic Bloch vectors $\set{\Psi_1(\bk), \ldots, \Psi_4(\bk)} \subset \mathcal{H} \oplus \mathcal{H}'$. Let then $\pi_1 \colon \mathcal{H} \oplus \mathcal{H}' \to \mathcal{H}$ be the projection on the first leg of the direct sum, and define
\[ \psi_a(\bk) := \pi_1 \Psi_a(\bk), \quad a \in \set{1, \ldots, 4}\,.\]
The collection $\set{\psi_1(\bk), \ldots, \psi_4(\bk), \widetilde{\psi}_2(\bk), \ldots, \widetilde{\psi}_m(\bk)}$ provides the desired Parseval frame for $P(\bk)$ containing $m+2$ Bloch vectors.
\end{proof}

\begin{remark}[On the definition of $Q_2$] \label{rmk:Q_2}
It would be tempting to put once again $Q_2(\bk):= \overline{P_2(-\bk)}$, as in the lower-dimensional cases, but this $Q_2$ has the \emph{opposite} first Chern class and the \emph{same} second Chern class with respect to the ones of $P_2$. An alternative construction of $Q_2(\bk)$ could be concocted as the eigenprojections of a lattice Dirac Hamiltonian $h(\bk) := \vec{d}(\bk) \cdot \vec{\Gamma}$, where $\vec{\Gamma} = (\Gamma^0, \ldots, \Gamma^4)$ is the vector of $4\times 4$ Dirac matrices. Ref.\ \onlinecite{qi2008topological} presents examples of such Hamiltonians which satisfy time-reversal symmetry (and hence \cite{panati2007triviality, monaco2015symmetry} have a vanishing first Chern class) but exhibit non-trivial values of the second Chern number, say equal to $-1$; any non-zero integer can be obtained by appropriately ``wrapping'' one of the directions of the 4D Brillouin torus. By adding time-reversal-symmetry breaking terms to such Hamiltonians, non-trivial values of the first Chern numbers can be obtained as well.

Alternatively, one can use the results of Ref.\ \onlinecite[Sec.~2.2]{cornean2016construction} to construct $Q_2$ as follows. It is possible to construct a rank-$2$-projection-valued map $\widetilde{P_2}(\bk)$ which enjoys the same smoothness and periodicity properties as $P_2(\bk)$, is unitarily equivalent to $P_2(\bk)$, and such that the ranges of $\widetilde{P_2}(\bk)$ lie in a \emph{$\bk$-independent} finite-dimensional space $\mathcal{H}' \simeq \C^N$; we can assume $N \ge 4$. The two families $\widetilde{P_2}(\bk)$ and $P_2(\bk)$ share the same Chern classes, as they are unitarily equivalent. On the other hand, since $\widetilde{P_2}(\bk) + [\Id_N - \widetilde{P_2}(\bk)] = \Id_N$ (where $\Id_N$ is the identity in the Hilbert space $\mathcal{H}'$), the family of projections $Q(\bk) := \Id_N - \widetilde{P_2}(\bk)$ has the opposite Chern classes, but has rank $N-2$. In view of the previous construction, we can ``squeeze'' the topology of $Q(\bk)$ in a sub-projection $Q_2(\bk)$ of rank $2$ with the same first and second Chern classes as $Q(\bk)$. The latter provides the required family of projections as in \eqref{eqn:Q_2}.
\end{remark}

\begin{acknowledgments}
D.~M.\ gratefully acknowledges the financial support from the National Group of Mathematical Physics (GNFM--INdAM) within the project Progetto Giovani GNFM 2020, as well as from ``Sapienza'' University of Rome within the projects Progetti di Ricerca Medi 2020 and 2021.
\end{acknowledgments}

\section*{Data Availability Statement}

Data sharing is not applicable to this article as no new data were created or analyzed in this study.

\appendix

\section{Properties of the $3$-degree} \label{app:3deg}

In this Appendix we collect a few key properties of the $3$-degree for $SU(m)$-valued maps, defined in \eqref{eqn:3deg}. More properties of this topological degree (including the fact that it is integer-valued) can be found in Ref.s\ \onlinecite{mukherjee2015differential, manton2004topological}.

Here and in the next Appendix we make use of the following \emph{graded cyclicity of the trace} for matrix-valued differential forms: if $\omega$ and $\eta$ are differential forms with coefficients in $m \times m$ matrices of degree $p$ and $q$ respectively, then
\begin{equation} \label{eqn:graded_cyclic}
\Tr_{\C^m} \left( \omega \wedge \eta \right) = \sum_{I \in \N^p, \: J \in \N^q} \Tr_{\C^m} \left( \omega_I \, \eta_J \right) \di \bk_I \wedge \di \bk_J = \sum_{I \in \N^p, \: J \in \N^q} \Tr_{\C^m} \left( \eta_J \, \omega_I \right) \cdot (-1)^{p\,q} \di \bk_J \wedge \di \bk_I = (-1)^{p\,q} \, \Tr_{\C^m} \left( \eta \wedge \omega \right)
\end{equation}
where, if $I = \set{i_1, \ldots, i_p} \in \N^p$, we have denoted $\di \bk_I := \di k_{i_1} \wedge \cdots \wedge \di k_{i_p}$. In particular, if $p$ is odd, it follows from the above identity that
\begin{equation} \label{eqn:power_odd}
\Tr_{\C^m}( \omega \wedge \omega) = (-1)^{p^2} \Tr_{\C^m}(\omega \wedge \omega) \quad \Longrightarrow \quad \Tr_{\C^m}( \omega \wedge \omega) = 0, \quad \text{and more in general} \quad \Tr_{\C^m}( \omega^{\wedge 2r} ) = 0 \text{ for all } r \in \N\,.
\end{equation}

\begin{proposition} \label{prop:3deg}
Let $\sigma_0(\bk), \sigma_1(\bk)$ be $m \times m$ matrices depending smoothly on $\bk = (k_2,k_3,k_4) \in \R^3$. Then
\begin{equation} \label{eqn:additive3deg}
\Tr_{\C^m} \left\{ \left[ \left( \sigma_0 \, \sigma_1 \right)^{-1} \, \di \left( \sigma_0 \, \sigma_1 \right) \right]^{\wedge 3} \right\} = \Tr_{\C^m} \left[ \left( \sigma_0^{-1} \, \di \sigma_0 \right)^{\wedge 3} \right] + \Tr_{\C^m} \left[ \left( \sigma_1^{-1} \, \di \sigma_1 \right)^{\wedge 3} \right] - 3 \, \di \left[ \Tr_{\C^m} \left( \sigma_0^{-1} \, \di \sigma_0 \wedge \di \sigma_1 \, \sigma_1^{-1}\right) \right]\,.
\end{equation}

In particular, if $\sigma_0$ and $\sigma_1$ are normalized pseudo-periodic $SU(m)$-valued maps, then:
\begin{enumerate}
 \item for all $\bG \in \Gamma^* \simeq 2 \pi \Z^3$
 \begin{equation} \label{eqn:3deg_periodic}
 \int_{\Omega + \bG} \Tr_{\C^m} \left[ \left( \sigma_0^{-1} \, \di \sigma_0 \right)^{\wedge 3} \right] = \int_\Omega \Tr_{\C^m} \left[ \left( \sigma_0^{-1} \, \di \sigma_0 \right)^{\wedge 3} \right] \equiv \int_{\T^3} \Tr_{\C^m} \left[ \left( \sigma_0^{-1} \, \di \sigma_0 \right)^{\wedge 3} \right]
 \end{equation}
 where $\Omega = [0,2\pi]^3$ is a fundamental cell for $\Gamma^*$;
 \item the $3$-degree of normalized pseudo-periodic $SU(m)$-valued maps is additive:
 \begin{equation} \label{eqn:3deg_additive}
 \tdeg(\sigma_0 \, \sigma_1) = \tdeg(\sigma_0) + \tdeg(\sigma_1), \quad \text{for } \tdeg(\sigma) := \frac{1}{24\,\pi^2} \int_{\T^3} \Tr_{\C^m} \left[ \left( \sigma^{-1} \, \di \sigma \right)^{\wedge 3} \right]\,;
 \end{equation}
  \item \label{item:3deg_homotopy} if $\sigma_0$ can be continuously deformed into $\sigma_1$ via normalized pseudo-periodic $SU(m)$-valued maps, then
 \begin{equation} \label{eqn:3deg_homotopy}
 \tdeg(\sigma_0) = \tdeg(\sigma_1)\,,
 \end{equation}
 that is, the $3$-degree is an homotopy invariant within this class of maps and corresponding deformations.
\end{enumerate}
\end{proposition}
\begin{proof}
A long but straightforward computation, using the Leibniz rule for the exterior differential and the above graded cyclicity of the trace, shows that
\begin{align*}
\Tr_{\C^m} \left\{ \left[ \left( \sigma_0 \, \sigma_1 \right)^{-1} \, \di \left( \sigma_0 \, \sigma_1 \right) \right]^{\wedge 3} \right\} & = \Tr_{\C^m} \left[ \left( \sigma_0^{-1} \, \di \sigma_0 \right)^{\wedge 3} \right] + \Tr_{\C^m} \left[ \left( \sigma_1^{-1} \, \di \sigma_1 \right)^{\wedge 3} \right] \\
& \quad - 3 \, \Tr_{\C^m} \left( \di \sigma_0^{-1} \wedge \di \sigma_0 \wedge \di \sigma_1 \, \sigma_1^{-1} \right) - 3 \, \Tr_{\C^m} \left( \sigma_0^{-1} \di \sigma_0 \wedge \di \sigma_1 \wedge \di \sigma_1^{-1} \right).
\end{align*}
The identity \eqref{eqn:additive3deg} now follows upon observing that
\[ \di \sigma_0^{-1} \wedge \di \sigma_0 \wedge \di \sigma_1 \, \sigma_1^{-1} = \di \left( \sigma_0^{-1} \, \di \sigma_0 \wedge \di \sigma_1 \right) \, \sigma_1^{-1} = \di \left( \sigma_0^{-1} \, \di \sigma_0 \wedge \di \sigma_1 \, \sigma_1^{-1} \right) - \sigma_0^{-1} \, \di \sigma_0 \wedge \di \sigma_1 \wedge \di \sigma_1^{-1}. \]

From \eqref{eqn:additive3deg} it can now be deduced that
\begin{equation} \label{eqn:long}
\begin{aligned}
\Tr_{\C^m} \left\{ \left[ \left( \sigma_0^{-1} \, \sigma_1 \, \sigma_0 \right)^{-1} \, \di \left( \sigma_0^{-1} \, \sigma_1 \, \sigma_0 \right) \right]^{\wedge 3} \right\} & = \Tr_{\C^m} \left[ \left( \sigma_0 \, \di \sigma_0^{-1} \right)^{\wedge 3} \right] + \Tr_{\C^m} \left[ \left( \sigma_1^{-1} \, \di \sigma_1 \right)^{\wedge 3} \right] + \Tr_{\C^m} \left[ \left( \sigma_0^{-1} \, \di \sigma_0 \right)^{\wedge 3} \right] \\
& \quad - 3 \, \di \Tr_{\C^m} \left( \sigma_0 \, \di \sigma_0^{-1} \wedge \di \sigma_1 \, \sigma_1^{-1} \right) - 3 \, \di \Tr_{\C^m} \left( \sigma_1^{-1} \di \sigma_1 \wedge \di \sigma_0 \, \sigma_0^{-1} \right) \\
& \quad - 3 \, \di \Tr_{\C^m} \left( \sigma_1^{-1} \, \sigma_0 \, \di \sigma_0^{-1} \, \sigma_1  \wedge \di \sigma_0 \, \sigma_0^{-1} \right)\,.
\end{aligned}
\end{equation}
Observe now that trivially
\[ \sigma \, \sigma^{-1} \equiv \Id \quad \Longrightarrow \quad \di \sigma \, \sigma^{-1} + \sigma \, \di \sigma^{-1} = 0, \]
and therefore the following identities hold:
\begin{align*}
\Tr_{\C^m} \left[ \left( \sigma_0 \, \di \sigma_0^{-1} \right)^{\wedge 3} \right] & = \Tr_{\C^m} \left[ \left( - \di \sigma_0 \, \sigma_0^{-1} \right)^{\wedge 3} \right] = - \Tr_{\C^m} \left[ \left( \sigma_0^{-1} \, \di \sigma_0 \right)^{\wedge 3} \right], \\
\Tr_{\C^m} \left( \sigma_0 \, \di \sigma_0^{-1} \wedge \di \sigma_1 \, \sigma_1^{-1} \right) & = \Tr_{\C^m} \left[ \left(- \di \sigma_0 \, \sigma_0^{-1} \right) \wedge \left( \di \sigma_1 \, \sigma_1^{-1} \right) \right] = \Tr_{\C^m} \left( \di \sigma_1 \, \sigma_1^{-1} \wedge \di \sigma_0 \, \sigma_0^{-1} \right)\,.
\end{align*}
Therefore, \eqref{eqn:long} can be simplified to
\begin{equation} \label{eqn:3deg_conjugation}
\Tr_{\C^m} \left\{ \left[ \left( \sigma_0^{-1} \, \sigma_1 \, \sigma_0 \right)^{-1} \, \di \left( \sigma_0^{-1} \, \sigma_1 \, \sigma_0 \right) \right]^{\wedge 3} \right\} = \Tr_{\C^m} \left[ \left( \sigma_1^{-1} \, \di \sigma_1 \right)^{\wedge 3} \right] - 3 \, \di \Tr_{\C^m} \left[ \left( \di \sigma_1 \, \sigma_1^{-1} + \sigma_1^{-1} \, \di \sigma_1 - \sigma_1^{-1} \, \di \sigma_0 \, \sigma_0^{-1} \, \sigma_1 \right) \wedge \di \sigma_0 \, \sigma_0^{-1} \right]\,.
\end{equation}

Let now $\sigma(\bk) \in SU(m)$ be pseudo-periodic. Let $\iota \colon \R^3 \to \R^3$ be the map that shifts any of the three coordinates by $2\pi$, say $\iota(k_2,k_3,k_4) := (k_2+2\pi,k_3,k_4)$ for example. Then pseudo-periodicity of $\sigma$ means that
\[ \iota^* \, \sigma = \alpha_*^{-1} \, \sigma \, \alpha_*, \]
where $\iota^* \, \sigma := \sigma \circ \iota$ and $\alpha_* \in \set{\Id_m, \alpha\sub{2D}, \alpha\sub{3D}}$ is selected depending on which coordinate is shifted (compare Lemma \ref{lemma:4Dalpha}). According to \eqref{eqn:3deg_conjugation}, we have corrispondingly that
\[ \iota^* \Tr_{\C^m} \left[ \left( \sigma^{-1} \, \di \sigma \right)^{\wedge 3} \right] = \Tr_{\C^m} \left[ \left( \sigma^{-1} \, \di \sigma \right)^{\wedge 3} \right] - 3 \, \di \Tr_{\C^m} \left[ \left( \di \sigma \, \sigma^{-1} + \sigma^{-1} \, \di \sigma - \sigma^{-1} \, \di \alpha_* \, \alpha_*^{-1} \, \sigma \right) \wedge \di \alpha_* \, \alpha_*^{-1} \right]\,. \]
Integrate now both sides on $\Omega = [0,2\pi]^3$ and apply Stokes' theorem, to deduce that
\begin{align*}
\int_{\iota(\Omega)} \Tr_{\C^m} \left[ \left( \sigma^{-1} \, \di \sigma \right)^{\wedge 3} \right] & = \int_{\Omega} \iota^* \Tr_{\C^m} \left[ \left( \sigma^{-1} \, \di \sigma \right)^{\wedge 3} \right] \\
& = \int_{\Omega} \Tr_{\C^m} \left[ \left( \sigma \, \di \sigma^{-1} \right)^{\wedge 3} \right] - 3 \, \int_{\partial \Omega} \Tr_{\C^m} \left[ \left( \di \sigma \, \sigma^{-1} + \sigma^{-1} \, \di \sigma - \sigma^{-1} \, \di \alpha_* \, \alpha_*^{-1} \, \sigma \right) \wedge \di \alpha_* \, \alpha_*^{-1} \right]\,.
\end{align*}

Focus now on the boundary term on the right-hand side of the above identity. If $\sigma$ further satisfies the normalization \eqref{eqn:boundary_sigma}, then $\sigma \equiv \Id_m$ on $\partial \Omega$ and correspondingly $\di \sigma \equiv 0$ there. Moreover, observe that $\alpha_*$ is of the form
\[ \alpha_* = \begin{pmatrix} \eu^{\iu \, \vec{n} \cdot \vec{k}} & 0 \\ 0 & \Id_{m-1} \end{pmatrix} \quad \Longrightarrow \quad \di \alpha_* \, \alpha_*^{-1} = \iu \, \vec{n} \cdot \di \vec{k} \, \begin{pmatrix} 1 & 0 \\ 0 & 0_{m-1} \end{pmatrix}, \]
where $\vec{n} = (n, m) \in \Z^2$ is a vector with integer entries (equal to $0$ or to an appropriate first Chern number of the underlying family of projections) and $\vec{k}$ is a vector of two $k$-coordinates for some face on $\partial \Omega$. It can then be computed that 
\[ \di \alpha_* \, \alpha_*^{-1} \, \sigma \wedge \di \alpha_* \, \alpha_*^{-1} = - (\vec{n} \cdot \di \vec{k})^{\wedge 2} \begin{pmatrix} \sigma_{11} & 0 \\ 0 & 0_{m-1} \end{pmatrix} = 0_m \]
in view of the skew-symmetry of the wedge product on $1$-forms, which implies $(\vec{n} \cdot \di \vec{k})^{\wedge 2} = 0$. In conclusion
\[ \int_{\iota(\Omega)} \Tr_{\C^m} \left[ \left( \sigma^{-1} \, \di \sigma \right)^{\wedge 3} \right] = \int_{\Omega} \Tr_{\C^m} \left[ \left( \sigma^{-1} \, \di \sigma \right)^{\wedge 3} \right] \]
from which \eqref{eqn:3deg_periodic} immediately follows. A similar argument, combining \eqref{eqn:additive3deg} and Stokes' theorem, also yields \eqref{eqn:3deg_additive}.

It remains to prove \eqref{eqn:3deg_homotopy}. Let therefore $\sigma_s$, $s \in [0,1]$, be a normalized pseudo-periodic homotopy between $\sigma_0$ and $\sigma_1$; without loss of generality, we assume that it depends smoothly on $s$ as well. We compute the derivative with respect to $s$ of the form $(\sigma_s^{-1} \, \di \sigma_s)^{\wedge 3}$: the derivative can ``hit'' any of the three factors in the wedge product, but up to graded cyclicity the three corresponding summands will yield the same result. Therefore
\begin{align*}
\partial_s \, \Tr_{\C^m} \left[ \left( \sigma_s^{-1} \, \di \sigma_s \right)^{\wedge 3} \right] & = 3 \, \Tr_{\C^m} \left[ \partial_s \left( \sigma_s^{-1} \, \di \sigma_s \right) \wedge \left( \sigma_s^{-1} \, \di \sigma_s \right)^{\wedge 2} \right] = 3 \, \Tr_{\C^m} \left[ \left(\partial_s \, \sigma_s^{-1} \, \di \sigma_s + \sigma_s^{-1} \, \di \partial_s \sigma_s \right) \wedge \left( \sigma_s^{-1} \, \di \sigma_s \right)^{\wedge 2} \right] \\
& = 3 \, \Tr_{\C^m} \left[ \left(-\sigma_s^{-1} \, \partial_s \, \sigma_s \, \sigma_s^{-1} \, \di \sigma_s + \sigma_s^{-1} \, \di \partial_s \sigma_s \right) \wedge \left( \sigma_s^{-1} \, \di \sigma_s \right)^{\wedge 2} \right] \\
& = 3 \, \Tr_{\C^m} \left[ \left(- \partial_s \, \sigma_s \, \sigma_s^{-1} \, \di \sigma_s \, \sigma_s^{-1} + \di \partial_s \sigma_s \, \sigma_s^{-1} \right) \wedge \left( \di \sigma_s \, \sigma_s^{-1} \right)^{\wedge 2} \right] \\
& = 3 \, \Tr_{\C^m} \left[ \di \left(\partial_s \, \sigma_s \, \sigma_s^{-1} \right) \wedge \left( \di \sigma_s \, \sigma_s^{-1} \right)^{\wedge 2} \right] \\
& = 3 \, \di \, \Tr_{\C^m} \left[ \partial_s \, \sigma_s \, \sigma_s^{-1} \, \left( \di \sigma_s \, \sigma_s^{-1} \right)^{\wedge 2} \right] - 3 \, \Tr_{\C^m} \left[ \partial_s \, \sigma_s \, \sigma_s^{-1} \, \di \left( \di \sigma_s \, \sigma_s^{-1} \right)^{\wedge 2} \right]
\end{align*}
where we used the Leibniz rule for the exterior differential in the last equality. Let us now observe that
\begin{align*}
\di \left( \di \sigma_s \, \sigma_s^{-1} \right)^{\wedge 2} & = \di \left( \di \sigma_s \, \sigma_s^{-1} \right) \wedge \di \sigma_s \, \sigma_s^{-1} - \di \sigma_s \, \sigma_s^{-1} \wedge \di \left( \di \sigma_s \, \sigma_s^{-1} \right) = - \di \sigma_s \wedge \di \sigma_s^{-1} \wedge \di \sigma_s \, \sigma_s^{-1} + \di \sigma_s \, \sigma_s^{-1} \wedge \di \sigma_s \wedge \di \sigma_s^{-1} \\
& = \left(\di \sigma_s \, \sigma_s^{-1} \right)^{\wedge 3} - \left(\di \sigma_s \, \sigma_s^{-1} \right)^{\wedge 3} = 0
\end{align*}
and therefore by Stokes' theorem
\[ \partial_s \, \int_{\Omega} \Tr_{\C^m} \left[ \left( \sigma_s^{-1} \, \di \sigma_s \right)^{\wedge 3} \right] = 3 \, \int_{\partial \Omega} \Tr_{\C^m} \left[ \partial_s \, \sigma_s \, \sigma_s^{-1} \, \left( \di \sigma_s \, \sigma_s^{-1} \right)^{\wedge 2} \right]. \]
Since we assumed that $\sigma_s \equiv \Id_m$ on $\partial \Omega$ for all $s \in [0,1]$, we have that $\di \sigma_s = 0$ on the boundary of $\Omega$, and the conclusion follows.
\end{proof}

\begin{corollary} \label{cor:3deg(sigma)=3deg(eta)}
Let $\sigma(\bk) \in SU(m)$ and $\eta(\bk) \in SU(2)$ be as in the statement of Theorem \ref{thm:SU(m)->SU(2)}, and let also $\alpha(\bk) = \delta\sub{4D}(\bk) \, \sigma(\bk) \in U(m)$ with $\delta\sub{4D}$ as in \eqref{eqn:delta4D}. Then
\[ \frac{1}{24\,\pi^2} \int_{\T^3} \Tr_{\C^m} \left[ \left( \alpha^{-1} \, \di \alpha \right)^{\wedge 3} \right] = \frac{1}{24\,\pi^2} \int_{\T^3} \Tr_{\C^m} \left[ \left( \sigma^{-1} \, \di \sigma \right)^{\wedge 3} \right] = \frac{1}{24\,\pi^2} \int_{\T^3} \Tr_{\C^2} \left[ \left( \eta^{-1} \, \di \eta \right)^{\wedge 3} \right] \: \in \Z \, . \]
\end{corollary}
\begin{proof}
The relation \eqref{eqn:additive3deg} implies that
\[ \int_{\T^3} \Tr_{\C^m} \left[ \left( \alpha^{-1} \, \di \alpha \right)^{\wedge 3} \right] = \int_{\T^3} \Tr_{\C^m} \left[ \left( \delta\sub{4D}^{-1} \, \di \delta\sub{4D} \right)^{\wedge 3} \right] + \int_{\T^3} \Tr_{\C^m} \left[ \left( \sigma^{-1} \, \di \sigma \right)^{\wedge 3} \right]\,. \]
It can be easily verified that, for $\delta\sub{4D}$ as in \eqref{eqn:delta4D}, one has
\[ \Tr_{\C^m} \left[ \left( \delta\sub{4D}^{-1} \, \di \delta\sub{4D} \right)^{\wedge 3} \right] = - \iu \, \left( n_2 \, \di k_2 + n_3 \, \di k_3 + n_4 \, \di k_4 \right)^{\wedge 3} = 0 \]
in view of the skew-symmetry of the wedge product. This implies the first equality in the statement.

As for the second claim, Proposition \ref{prop:3deg}.\ref{item:3deg_homotopy} implies that $\tdeg(\sigma) = \tdeg(\widetilde{\sigma})$ where $\sigma, \widetilde{\sigma} \in SU(m)$ are in \eqref{eqn:sigma_eta}. It is clear that the constant diagonal block $\Id_{m-2}$ in the expression for $\widetilde{\sigma}(\bk)$ does not contribute to the integral that defines the $3$-degree, and therefore the trace is reduced to the $2 \times 2$ block which contains $\eta(\bk)$.
\end{proof}

\section{Chern forms and their primitives} \label{app:CS}

This Appendix is devoted to list some of the properties of the Chern forms, which are derived from the definition of the Berry curvature \eqref{eqn:Berry_curv} and of the Berry connection \eqref{eqn:BerryConnection} \cite{milnor1974characteristic, manton2004topological, nash2011topology, monaco2017chern}. We let $P(\bk)$, $\bk \in \T^d$, be a family of projections labeled by a $d$-dimensional torus, with coordinates $(k_1, \ldots, k_d)$. In the following we will denote $\partial_{\mu} \equiv \partial_{k_{\mu}}$ for $\mu \in \set{1,\ldots,d}$. We also use Einstein's convention that repeated indices should be summed over: Latin indices $a,b,c,\ldots$ range from $1$ to $m$ (the rank of the projection); Greek indices $\mu, \nu,\ldots,$ range from $1$ to $d$ (the dimension of the torus). In order not to overburden the notation, we also abbreviate $\Tr \equiv \Tr_{\C^m}$.

\begin{lemma}
Let $F_{\mu \nu}(\bk)$ be as in \eqref{eqn:Berry_curv}, and let $\set{\phi_a(\bk)}_{1 \le a \le m}$ be an orthonormal basis of Bloch functions for $P(\bk)$, $\bk \in \T^d$. Then
\begin{equation} \label{eqn:Fmunu}
F_{\mu \nu}(\bk)_{ab} = \frac{1}{2\pi\iu} \left[ \left\langle \partial_\mu \phi_a(\bk), \, \partial_\nu \phi_b(\bk) \right\rangle + \left\langle \phi_a(\bk), \, \partial_\mu \phi_c(\bk) \right\rangle \left\langle \phi_c(\bk), \, \partial_\nu \phi_b(\bk) \right\rangle - \big( \mu \leftrightarrow \nu \big) \right]\,.
\end{equation}
\end{lemma}
\begin{proof}
Dropping the dependence on $\bk$, we have that
\[ P = \left| \phi_c \right\rangle \left\langle \phi_c \right|, \quad \left\langle \phi_c, \, \phi_d \right\rangle = \delta_{cd}, \]
and therefore by definition
\begin{align*}
\big(F_{\mu \nu}\big)_{ab} & = \frac{1}{2\pi\iu} \left\langle \phi_a, \, P \left[ \partial_\mu P , \partial_\nu P \right] P \, \phi_b \right\rangle = \frac{1}{2\pi\iu} \left\langle  \phi_a, \, \partial_\mu P \, \partial_\nu P \, \phi_b \right\rangle - \big( \mu \leftrightarrow \nu \big) \\
& = \frac{1}{2\pi\iu} \left[ \left\langle \phi_a, \, \partial_\mu \phi_c \right\rangle \left\langle \phi_c, \,  \partial_\nu \phi_d \right\rangle \left\langle \phi_d, \, \phi_b \right\rangle + \left\langle \phi_a, \, \phi_c \right\rangle \left\langle \partial_\mu \phi_c, \,  \partial_\nu \phi_d \right\rangle \left\langle \phi_d, \, \phi_b \right\rangle \right. \\
& \qquad\qquad + \left. \left\langle \phi_a, \, \partial_\mu \phi_c \right\rangle \left\langle \phi_c, \, \phi_d \right\rangle \left\langle \partial_\nu \phi_d, \, \phi_b \right\rangle + \left\langle \phi_a, \, \phi_c \right\rangle \left\langle \partial_\mu \phi_c, \,  \phi_d \right\rangle \left\langle \partial_\nu \phi_d, \, \phi_b \right\rangle  - \big( \mu \leftrightarrow \nu \big) \right] \\
& = \frac{1}{2\pi\iu} \left[ \left\langle \phi_a, \, \partial_\mu \phi_c \right\rangle \, \big( \left\langle \phi_c, \, \partial_\nu \phi_b \right\rangle + \left\langle \partial_\nu \phi_c, \, \phi_b \right\rangle \big) + \left\langle \partial_\mu \phi_a, \,  \partial_\nu \phi_b \right\rangle + \left\langle \partial_\mu \phi_a, \, \phi_c \right\rangle \left\langle \partial_\nu \phi_c, \, \phi_b \right\rangle - \big( \mu \leftrightarrow \nu \big) \right]\,.
\end{align*}
Notice now that the term in round parentheses computes
\[ \left\langle \phi_c, \, \partial_\nu \phi_b \right\rangle + \left\langle \partial_\nu \phi_c, \, \phi_b \right\rangle = \partial_\nu \left\langle \phi_c , \, \phi_b \right\rangle = 0 \]
and that analogously 
\[ \left\langle \partial_\mu \phi_a, \, \phi_c \right\rangle \left\langle \partial_\nu \phi_c, \, \phi_b \right\rangle = \left\langle \phi_a, \, \partial_\mu \phi_c \right\rangle \left\langle \phi_c, \, \partial_\nu \phi_b \right\rangle \, . \]
This concludes the proof.
\end{proof}

\begin{corollary}
With $F$ as in \eqref{eqn:Berry_curv} and $A$ as in \eqref{eqn:BerryConnection}, we have
\begin{equation} \label{eqn:F=dA+AA}
F = \di A + 2 \pi \iu \, A \wedge A
\end{equation}
and the \emph{Bianchi identity}
\begin{equation} \label{eqn:Bianchi}
\di F = 2 \pi \iu \left(F \wedge A - A \wedge F \right)\,.
\end{equation}
\end{corollary}
\begin{proof}
Let us compute
\begin{align*}
\big(\di A\big)_{ab} & = \big(\partial_\mu A_\nu\big)_{ab} \, \di k_\mu \wedge \di k_\nu = \frac{1}{2} \sum_{\mu < \nu} \left( \partial_\mu A_\nu - \partial_\nu A_\mu  \right)_{ab} \, \di k_\mu \wedge \di k_\nu \\
& = \frac{1}{2} \sum_{\mu < \nu} \frac{1}{2\pi\iu} \left[ \left\langle \partial_\mu \phi_a, \, \partial_\nu \phi_b \right\rangle + \left\langle \phi_a, \, \partial_\mu \partial_\nu \phi_b \right\rangle - \left\langle \partial_\nu \phi_a, \, \partial_\mu \phi_b \right\rangle - \left\langle \phi_a, \,  \partial_\nu \partial_\mu \phi_b \right\rangle \right] \, \di k_\mu \wedge \di k_\nu \\
& = \frac{1}{2} \sum_{\mu < \nu} \frac{1}{2\pi\iu} \left[ \left\langle \partial_\mu \phi_a, \, \partial_\nu \phi_b \right\rangle - \big( \mu \leftrightarrow \nu \big) \right] \, \di k_\mu \wedge \di k_\nu\, , \\
\big(A \wedge A\big)_{ab} & = \big(A_\mu A_\nu\big)_{ab} \, \di k_\mu \wedge \di k_\nu = \frac{1}{2} \sum_{\mu < \nu} \big(\left[ A_\mu,\, A_\nu\right]\big)_{ab} \, \di k_\mu \wedge \di k_\nu \\
& = \frac{1}{2} \sum_{\mu < \nu} \left(\frac{1}{2\pi\iu}\right)^2 \, \left[ \left\langle \phi_a, \, \partial_\mu \phi_c \right\rangle \left\langle \phi_c, \, \partial_\nu \phi_b \right\rangle - \big( \mu \leftrightarrow \nu \big) \right] \, \di k_\mu \wedge \di k_\nu.
\end{align*}
By comparing the above identities with \eqref{eqn:Berry_curv} and \eqref{eqn:Fmunu}, the equality \eqref{eqn:F=dA+AA} follows at once. The latter also implies the Bianchi identity:
\begin{align*}
\di F & = \di \left( \di A + 2 \pi \iu \, A \wedge A \right) = 2 \pi \iu \, \left( \di A \wedge A - A \wedge \di A \right) = 2 \pi \iu \, \left[ \left( F - 2 \pi \iu \, A \wedge A \right) \wedge A - A \wedge \left( F - 2 \pi \iu \, A \wedge A \right) \right] \\
& = 2 \pi \iu \left(F \wedge A - A \wedge F \right)\,. \qedhere
\end{align*}
\end{proof}

\begin{proposition} \label{prop:CS}
Define the \emph{Chern--Simons $3$-form}
\begin{equation} \label{eqn:CS}
CS := \frac{1}{2} \left[ \Tr(A) \wedge \Tr(F) - \Tr \left( F \wedge A - \frac{2\pi\iu}{3} \, A \wedge A \wedge A \right) \right].
\end{equation}
Then
\begin{equation} \label{eqn:c2=dCS}
\di CS = c_2(P).
\end{equation}
Moreover, if the Bloch basis $\Psi$ is obtained from the basis $\Phi$ by a change of gauge $\gamma \in U(m)$, then
\begin{equation} \label{eqn:CS_gauge} 
CS^{(\Psi)} = CS^{(\Phi)} - \frac{1}{24\pi^2} \, \Tr \left[ \left( \gamma^{-1} \, \di \gamma \right)^{\wedge 3} \right] - \di B_\gamma, \quad \text{where} \quad B_\gamma := \frac{1}{4\pi\iu} \left[ \Tr\left(A \wedge \di \gamma \, \gamma^{-1}\right) + \Tr(A) \wedge \Tr\left(\di \gamma \, \gamma^{-1} \right) \right]\,.
\end{equation}
\end{proposition}
\begin{proof}
Let us first compute
\begin{align*}
\di \left[ \Tr(A) \wedge \Tr(F) \right] & = \Tr(\di A) \wedge \Tr(F) - \Tr(A) \wedge \Tr(\di F) = \Tr \left(\di A - 2\pi \iu \, A \wedge A \right) \wedge \Tr(F) - 2\pi\iu \, \Tr(A) \wedge \Tr(F \wedge A - A \wedge F) \\
& = \Tr(F) \wedge \Tr(F)\,.
\end{align*}
In the above chain of equalities, we have used that $\Tr(A \wedge A) = 0$ in view of \eqref{eqn:power_odd}, the Bianchi identity \eqref{eqn:Bianchi} for $\di F$, the identity \eqref{eqn:F=dA+AA}, and the graded cyclicity of the trace \eqref{eqn:graded_cyclic} to conclude that $\Tr(F \wedge A) = \Tr(A \wedge F)$. Next we can compute
\[ \di \, \Tr \left( F \wedge A - \frac{2\pi\iu}{3} \, A \wedge A \wedge A \right) = \Tr \left( \di F \wedge A + F \wedge \di A \right) - \frac{2\pi\iu}{3} \, \Tr \left( \di A \wedge A \wedge A - A \wedge \di A \wedge A + A \wedge A \wedge \di A\right)\,. \]
Notice that, in the last trace, all summands give the same contribution again by the graded cyclicity. Therefore we can proceed by applying \eqref{eqn:F=dA+AA} and \eqref{eqn:Bianchi} to get
\begin{align*}
\di \, \Tr \left( F \wedge A - \frac{2\pi\iu}{3} \, A \wedge A \wedge A \right) & = \Tr \left[ 2 \pi \iu \, F \wedge A \wedge A - 2 \pi \iu \, A \wedge F \wedge A + F \wedge \left( F - 2 \pi \iu \, A \wedge A \right) - 2\pi\iu \, A \wedge A \wedge \left(F - 2\pi\iu \, A \wedge A\right) \right] \\
& = \Tr \left[ 2 \pi \iu \, F \wedge A \wedge A + 2 \pi \iu \, A \wedge A \wedge F + F \wedge F - 2 \pi \iu \, F \wedge A \wedge A - 2\pi\iu \, A \wedge A \wedge F + (2\pi\iu)^2 \, A^{\wedge 4} \right] \\
& = \Tr(F \wedge F)\,,
\end{align*}
where we used the graded cyclicity of the trace to argue that $\Tr(A \wedge F \wedge A) = - \Tr(A \wedge A \wedge F)$ and \eqref{eqn:power_odd} to conclude that $\Tr(A^{\wedge 4})=0$. A comparison of the two identities obtained above with \eqref{eqn:c12} allows to conclude that \eqref{eqn:c2=dCS} holds.

Next we investigate the gauge dependence of the Chern--Simons form $CS$. Recall from \eqref{eqn:A_gauge} and the following equation how $A$ and $F$ change with a change of gauge $\gamma$ in the Bloch functions. With those, we are able to compute
\begin{align*}
CS^{(\Psi)} & = \frac{1}{2} \left\{ \Tr\left(\gamma^{-1} \, A \, \gamma + \frac{1}{2\pi\iu} \, \gamma^{-1} \, \di \gamma\right) \wedge \Tr\left(\gamma^{-1} \, F \, \gamma\right) - \Tr \left[ \gamma^{-1} \, F \, \gamma \wedge \left( \gamma^{-1} \, A \, \gamma + \frac{1}{2\pi\iu} \, \gamma^{-1} \, \di \gamma \right) \right] \right. \\
& \quad + \left. \frac{2\pi\iu}{3} \, \Tr \left[\left( \gamma^{-1} \, A \, \gamma + \frac{1}{2\pi\iu} \, \gamma^{-1} \, \di \gamma \right) \wedge \left( \gamma^{-1} \, A \, \gamma + \frac{1}{2\pi\iu} \, \gamma^{-1} \, \di \gamma \right) \wedge \left( \gamma^{-1} \, A \, \gamma + \frac{1}{2\pi\iu} \, \gamma^{-1} \, \di \gamma \right) \right] \right\} \\
& = CS^{(\Phi)} + \frac{1}{2} \cdot \frac{2\pi\iu}{3} \cdot \left(\frac{1}{2\pi\iu}\right)^3 \, \Tr \left[ \left(\gamma^{-1} \, \di \gamma \right)^{\wedge 3} \right] + \frac{1}{2} \cdot \frac{1}{2\pi\iu} \Tr\left( \gamma^{-1} \, \di \gamma\right) \wedge \Tr(F) \\
& \quad - \frac{1}{2} \left[ \frac{1}{2\pi\iu} \, \Tr \left( F \wedge \di \gamma \, \gamma^{-1} \right) - \Tr \left( A \wedge A \wedge \di \gamma \, \gamma^{-1} \right) - \frac{1}{2\pi\iu} \, \Tr \left( A \wedge \di \gamma \, \gamma^{-1} \wedge \di \gamma \, \gamma^{-1} \right) \right]
\end{align*}
where once again the graded cyclicity of the trace has been used repeatedly. In order to conclude the proof, let us notice that
\begin{align*}
\Tr\left( \gamma^{-1} \, \di \gamma\right) \wedge \Tr(F) & = - \Tr(\di A) \wedge \Tr\left( \gamma^{-1} \, \di \gamma\right) = - \di \left[ \Tr(A) \wedge \Tr\left( \gamma^{-1} \, \di \gamma\right) \right] - \Tr(A) \wedge \Tr\left( \di \gamma^{-1} \wedge \di \gamma\right) \\
& = - \di \left[ \Tr(A) \wedge \Tr\left( \gamma^{-1} \, \di \gamma\right) \right] + \Tr(A) \wedge \Tr \left[\left( \gamma^{-1} \wedge \di \gamma\right)^{\wedge 2} \right] = - \di \left[ \Tr(A) \wedge \Tr\left( \gamma^{-1} \, \di \gamma\right) \right]
\end{align*}
where the last equality is due to \eqref{eqn:power_odd}. Moreover, in view of \eqref{eqn:F=dA+AA},
\begin{align*}
\Tr \left[ \left( F - 2\pi\iu A \wedge A \right) \wedge \di \gamma \, \gamma^{-1} \right] - \Tr \left( A \wedge \di \gamma \, \gamma^{-1} \wedge \di \gamma \, \gamma^{-1} \right) & = \Tr \left( \di A \wedge \di \gamma \, \gamma^{-1} \right) + \Tr \left( A \wedge \di \gamma \wedge \di \gamma^{-1} \right) \\
& = \di \Tr \left( A \wedge \di \gamma \, \gamma^{-1} \right)\,.
\end{align*}
Combining all these identities together, \eqref{eqn:CS_gauge} follows.
\end{proof}

\bibliography{bibliographie}% Produces the bibliography via BibTeX.

\end{document}